\documentclass[a4paper,12pt]{article}

\usepackage{amsmath,amsfonts,amssymb}
\usepackage{theorem}
\usepackage{indentfirst}
\usepackage{graphicx}
\usepackage{hyperref,url}
\usepackage{longtable}
\usepackage[usenames,dvipsnames]{color}

\hypersetup{
    colorlinks=true,       % false: boxed links; true: colored links
    linkcolor=NavyBlue,          % color of internal links
    citecolor=blue,        % color of links to bibliography
    filecolor=magenta,      % color of file links
    urlcolor=cyan           % color of external links
}

\usepackage[sort&compress]{natbib}

\bibpunct{(}{)}{;}{a}{,}{,}
\renewcommand{\cite}{\citet}

\theoremstyle{plain} { \theorembodyfont{\rmfamily}
\newtheorem{definition}{Definition}
\newtheorem{remark}{Remark}
\newtheorem{example}{Example}
}
\newtheorem{proposition}{Proposition}
\newtheorem{lemma}{Lemma}
\newtheorem{corollary}{Corollary}
\newtheorem{theorem}{Theorem}

\newcommand{\qed}{\hfill \mbox{\raggedright \rule{.07in}{.1in}}}
\newenvironment{proof}{\vspace{1ex}\noindent{\bf Proof}\hspace{0.5em}}
    {\hfill\qed\vspace{1ex}}

\newcommand{\cF}{{\cal F}}
\newcommand{\cM}{{\cal M}}
\newcommand{\cP}{{\cal P}}

\newcommand{\cL}{{\cal L}}
\newcommand{\prob}{\mathbb{P}}
\newcommand{\expect}{\mathbb{E}}
\newcommand{\LP}{\mathbb{L}}
\newcommand{\quan}{{\cal Q}}
\newcommand{\cdf}{{\cal D}}
\newcommand{\dom}{\text{dom}}

\newcommand{\density}{\Psi}
\newcommand{\unitdensity}{\Phi}

\newcommand{\ba}{\begin{eqnarray}}
\newcommand{\ea}{\end{eqnarray}}
\newcommand{\baa}{\begin{eqnarray*}}
\newcommand{\eaa}{\end{eqnarray*}}
\newcommand{\bR}{\mathbb{R}}
\newcommand{\bP}{\mathbb{P}}

\newcommand{\cC}{\mathcal{C}}
\newcommand{\cD}{\mathcal{D}}
\newcommand{\rhoc}{\rho_\delta}
\newcommand{\re}{\widehat{\rho}}

\newcommand{\VaR}{\textrm{VaR}}
\newcommand{\SVaR}{\textrm{SVaR}}

\usepackage{setspace}
\title{Loss-based risk measures}
\author{Rama CONT$^{1}$, Romain DEGUEST $^{2}$ and Xue Dong HE$^3$\\
\ \\
1)\  Laboratoire de Probabilit\'es et Mod\`eles Al\'eatoires\\
CNRS- Universit\'e Pierre et Marie Curie, France.\\
2)\ EDHEC Business School, Nice (France).\\
3)\ IEOR Dept, Columbia University, New York.}
\date{2011. Revised: 2012. \\
To appear in: {\sc Statistics \& Risk Modeling}, 2013.} %\date{Financial Engineering Report No. 2011-03\\ Center for Financial Engineering, Columbia University.}
\begin{document}
\maketitle
\begin{abstract}
Starting from the requirement that risk of financial portfolios
should be measured in terms  of their losses, not their gains, we
define the notion of {\it loss-based risk measure} and study the
properties of this class of risk measures. We characterize convex
loss-based risk measures by a representation theorem and give
examples of such risk measures. We then discuss the statistical
robustness of the risk estimators associated with the family of
loss-based risk measures: we provide a general criterion for the
qualitative robustness of the risk estimators and
 compare this criterion with a sensitivity analysis of estimators based on
influence functions. We find that the risk estimators associated with convex loss-based risk measures are not robust.
\end{abstract}

%\tableofcontents
\newpage

\section{Introduction}
\subsection{Motivation}
A main focus  of quantitative modeling in finance has been to
measure the risk of financial portfolios. In connection with the
widespread use of Value-at-Risk (VaR) and related risk measurement
methodologies,  a considerable theoretical literature
\citep{acerbi2002smr,acerbi2007cmr,artzner1999cmr,contdeguestscandolo,FollmerHSchiedA:02convexriskmeasure,follmer2011sfi,frittelli2002por,heyde2006grm,mcneil2005qrm}
has focused on the design of appropriate risk measures for financial
portfolios. In this   approach,   a risk measure is represented as a
real-valued map   assigning to each  random variable
$X$---representing the payoff of a portfolio---a number which
measures the risk of this portfolio. A framework often used as a
starting point is the axiomatic setting of { \cite{artzner1999cmr},
which defines a {\it coherent risk measure} as a map
$\rho:L^{\infty}(\Omega,{\cal F},\mathbb{P})\to\mathbb{R}$ that is
\begin{enumerate}
    \item \textit{monotone} (\emph{decreasing}): $\rho(X)\leq
\rho(Y)$ provided $X\geq Y$;
    \item \textit{cash-additive (additive with respect to cash reserves)}:
$\rho(X+c)= \rho(X)-c$ for any $c\in\mathbb{R}$;
    \item \textit{positive homogeneous}: $\rho(\lambda X)=\lambda\rho(X)$ for any $\lambda\geq 0$;
    \item \textit{sub-additive}: $\rho(X+Y)\leq \rho(X)+\rho(Y)$.
\end{enumerate}
\cite{artzner1999cmr} argue that these axioms correspond to
desirable properties of a risk measure, such as the reduction of
risk under diversification. These axioms have provided   an elegant
mathematical framework for the study of coherent risk measures, but
fail to take into account some key features encountered in the
practice of risk management, as illustrated by the following
(important) example.

Consider a central clearing facility or an exchange, in which
various market participants  clear portfolios of financial
instruments. Any   participant of the
 clearing house must deposit a margin
requirement   for the purpose of covering the potential cost of
liquidating the   clearing participant's portfolio in case of
default. The risk measurement problem facing the exchange is
therefore to determine the margin requirement for each portfolio,
which is in this case the risk measure of the portfolio as seen by
the exchange. Unlike the situation of an investor evaluating his/her
own risk, the exchange is affected by the gains and losses  of the
market participants in an asymmetric way. As long as the market
participant's positions results in a gain, the gain is kept by the
participant, but if the participant suffers a loss, the exchange may
have to step in and cover the loss in case the participant default.
It follows that, when measuring the risk posed {\it to the exchange}
by a participant's portfolio, it is only relevant to consider the
{\it losses} of this portfolio, not the gains. Indeed, the well
known Standard Portfolio ANalysis (SPAN) method introduced by the
Chicago Merchantile Exchange and used by many other exchanges,
computes the margin requirement of a financial portfolio with profit
and loss (P\&L) $X$   as the maximum {\em loss} of the portfolio
over a set of pre-selected stress scenarios
$\omega_1,\dots,\omega_n$:
\begin{align}\label{eq:SPAN}
\rho(X)= \max \{ - \min(X(\omega_1),0),..., - \min(X(\omega_n),0)\}.
\end{align}
As we can see in this example, the risk measure of a portfolio is
only based on the loss $\min(X,0)$ that is, the negative part of
$X$.

The argument that a risk measure should be based on losses, not
gains, is not restricted to the problem of computing margin
requirements for a central clearing facility. Indeed, a regulator
faces a similar issue when evaluating the cost of a bank failure:
these costs materialize  only in scenarios when a bank undergoes
large losses resulting in its default, whereas the trading gains of
a bank do not positively affect the regulator's position. Thus, the
risk of a bank's portfolio, as viewed by the regulator, should also
be based on the magnitude of the bank's {\it loss}, not its
potential gains.

These examples show that, a risk measure    used for determining
capital (or margin) requirements (called 'external' risk measure in
\cite{heyde2006grm})  should   be solely based on the loss of a
portfolio. This property can be formulated by requiring the risk
measure $\rho(X)$  to depend only on the negative part of $X$,
representing the loss:
\begin{equation} \rho(X) = \rho(\  \min(X,0)\ ). \label{eq.loss}\end{equation}

% Even for the purpose of internal risk management inside a bank, it is
% reasonable to evaluate risk by focusing on losses. After all, the
% main function of risk management teams is to control the damage
% caused by huge losses rather than to achieve substantial gains.

This property is clearly not contained in the axioms of coherent
risk measures. In fact, the cash-additivity property   implies that
coherent risk measures must depend on gains as well, which clearly
contradicts \eqref{eq.loss}.  So, one may not simply add the
loss-dependence property \eqref{eq.loss} to the axioms of coherent
risk measures without reconsidering the other axioms. In fact, the
CME SPAN method does {\it not} verify the cash-additivity axiom and
therefore is not a coherent risk measure.\footnote{It is interesting
to note that the CME SPAN method was one of the motivations cited in
\cite{artzner1999cmr} for introducing the framework of 'coherent'
risk measures.}

Many revisions to the axioms of coherent risk measures have been
proposed and studied in the literature, replacing in particular
positive homogeneity and sub-additivity   with the more general
convexity property
\citep{FollmerHSchiedA:02convexriskmeasure,follmer2011sfi,frittelli2002por},
co-monotonic sub-additivity  \citep{heyde2006grm} or  co-monotonic
convexity \citep{SongYan:2009riskmeasure}. But these alternative
frameworks still rely on cash-additivity and do not consider the
property of loss-dependence as formulated in \eqref{eq.loss}, so a
proper definition of loss-based risk measure calls for a systematic
revision of the axioms of \cite{artzner1999cmr} along new
directions.

 Let us note here
that the property of loss-dependence is not exactly the same as
requiring the risk measure to depend on, say, the left tail of the
loss distribution (which is the case for Value at Risk or Expected
Shortfall, for instance). The example of a portfolio with random,
but positive payoffs shows the difference between these two.

We propose a new class of risk measures,   {\em loss-based risk
measures}, which depend only on the loss  of a  financial portfolio,
and investigate the properties of such risk measures.

Since  the cash-additivity property is incompatible with the
loss-based property, we remove the cash-additivity for risk
measures. However, it is worth mentioning that  loss-based
risk measures, though not   cash-additive,
do not necessarily violate the property 
$\rho(X+\rho(X))=0$. Indeed, it is easy to verify that the
  CME SPAN method satisfies this
property,  without verifying the cash-additivity property.
\cite{ElKarouiRavanelli:2009CashSubadditive} challenge the axiom of
cash additivity from a different angle, showing that when
considering risk measures defined on {\em future} (instead of
discounted) value of portfolio gains in order to take into account
interest rate risk,  cash additivity is inevitably violated and
needs to be replaced by cash subadditivity. Our approach provides an
independent motivation for relaxing cash-additivity as a property
for risk measures.

Another issue  which is extremely important in practice but somewhat
neglected in the literature on risk measures is the issue of
statistical estimation   and the design of robust estimators for
risk measures. Risk measures are usually defined in terms of a portfolio's
profit/loss (P\&L) or their distributions. Because these distributions are
not directly observed, one has to use historical data to {\it estimate } the risk of
each portfolio. For instance, one may use the historical P\&Ls of a
portfolio to estimate the distribution of the portfolio P\&L and
apply a risk measure on the estimated distribution to obtain an
estimate of the portfolio risk. Such a procedure, which starts
from historical data or simulated  losses as input and obtains the estimated
risk of the portfolio as output, is called a risk estimator. Roughly
speaking, a risk estimator is robust if the resulting risk estimate
of a portfolio is not extremely sensitive to a small change in the
sample. A non-robust risk estimator may vary dramatically from day
to day which makes it difficult to use and even more difficult,  if not impossible, to backtest. For
instance, if the clearing house computes margin requirements
according to a non-robust risk estimator,  market participants may
find the resulting margin requirements to be extremely volatile and
thus unacceptable. Unfortunately, as shown in
\cite{contdeguestscandolo}, an unintended consequence of
subadditivity property    is that it
 requires dependence on extreme tail events,
which leads in turn to high sensitivity to outliers and lack of
robustness. We study these issues in the context of loss-based risk
measures, extending  previous work of
\cite{contdeguestscandolo}, and characterize loss-based
risk measures which admit robust estimators.

Loss-based risk measures have in fact a long tradition in actuarial science: \cite{Hattendorff1868} is an early example. In the financial risk management context, this idea was explored by
\cite{jarrow02} and, in parallel with the present work, by \cite{Staum:2011ExcessInvarianceShortfall} in   a discrete setting.
\cite{jarrow02} defines a risk measure that is the premium of the
put option on a portfolio's net value. 
 Our setting extends these examples, discussed in Section 2, and allows to
obtain a characterization of loss-based risk measures on a
general probability space and study the robustness of the associated
risk estimators.

\subsection{Main results}
We consider in this paper an alternative approach to defining risk
measures which addresses these concerns.
 Starting from the requirement that risk measures of financial
portfolios should be based on their losses, not their gains, we
define  the notion of {\it loss-based risk measure} and study the
properties of this class of risk measures.

We first provide a dual representation for convex loss-based risk
measures, which are loss-based risk measures satisfying a convexity
property. This representation is similar to that of convex risk
measures, and  states that a convex loss-based risk measure is
worst-case expected loss (adjusted by some penalty). For {\it
statistical} convex loss-based risk measures  i.e. which only depend
on the loss distribution, we provide another representation theorem
in terms of portfolio loss quantiles.

We provide abundant examples of convex loss-based risk measures,
many of which are obtained by simply replacing P\&Ls with their loss
parts in certain convex risk measures. However, we also provide an
example that cannot be constructed from convex risk measures in this
way. This example illustrates that convex loss-based risk measures
are not trivial extensions of convex risk measures. We further prove
that a convex loss-based risk measure can be constructed from a
convex risk measure by replacing P\&Ls with their loss parts if and
only if it satisfies a property which we call cash-loss additivity.

We then investigate the robustness of the risk estimators associated
with a family of statistical loss-based risk measures that include
both statistical convex loss-based risk measures and VaR on losses
as special cases. Using a notion of robustness for risk estimators
given in \cite{contdeguestscandolo}, we give  a necessary and
sufficient  condition for the risk estimators to be robust. Our
results imply that risk estimators associated with {\it convex } statistical
loss-based risk measures are {\em not} robust, whereas sample
loss quantiles are robust. These conclusions are further confirmed
by investigating the influence function of a large class of
statistical loss-based risk measures.

One of our main results is that   the convexity property, which
leads to reduction of risk under diversification, cannot coexist
with robustness. Therefore, when choosing a risk measure, one has to
decide which property is more important, and the choice should be
dependent of the context in which the risk measure is used. For
example, if the risk measure is used to compute margin requirements
frequently, which is the case in a clearing house, robustness might
be more important than convexity. If the risk is used for asset
allocation, then robustness may not  be the primary  issue  st stake
and convexity might be more relevant property in this case.

The  paper is organized as follows. In Section \ref{sec:lossbased}
we define loss-based risk measures and provide the representation
theorems for convex loss-based risk measures before we provide
several examples of loss-based risk measures. Section
\ref{sec:robustness} is devoted to studying the qualitative
robustness of the risk estimators associated with a family of
statistical loss-based risk measures. In Section
\ref{sec:sensitivity} we perform sensitivity analysis on a set of
loss-based risk measures and investigate their influence functions.
Finally, Section \ref{se:Conclusions} concludes the paper.

\section{Loss-Based Risk Measures}\label{sec:lossbased}

%\subsection{Notations}

Consider an atomless probability space $(\Omega, \cF,\prob)$
representing market scenarios. For a random variable $X$, denote by
$F_X(\cdot)$ its cumulative distribution function and $G_X(\cdot)$
its left-continuous quantile function. For any $p\in[1,\infty)$, let
$\LP^p(\Omega, \cF,\prob)$ be the space of random variables $X$ with
the norm $\|X\|_p:=\left(\expect[|X|^p]\right)^\frac{1}{p}$, and let
$\LP^\infty(\Omega, \cF,\prob)$ be the space of  bounded random
variables.
Let
\begin{align*}
  \qquad\cP(\Omega,\cF,\bP):=\{X\in \LP^1(\Omega, \cF,\prob):X\ge 0,||X||_1=1 \},
\end{align*}
\begin{align*}
{\rm and}\qquad  \cM(\Omega,\cF,\bP):=\{X\in \LP^1(\Omega, \cF,\prob):X\ge 0,||X||_1\le1 \}.
\end{align*}
Then, $\cP(\Omega,\cF,\bP)$ can also be regarded as the set of $\bP$-absolutely continuous probability measures on $(\Omega,\cF,\bP)$, and $\cM(\Omega,\cF,\bP)$ as the set of $\bP$-absolutely continuous measures $\mu$ such that $\mu(\Omega)\leq1$.

Denote by $\cP((0,1))$ the set of probability measures on the open
unit interval $(0,1)$, and $\cM((0,1))$ the set of  positive
measures $\mu$ on $(0,1)$ such that $\mu((0,1))\leq1$.  Let
\begin{align}
  \density((0,1)) :&=\left\{ \phi:(0,1)\mapsto\mathbb{R}_+ \mid \phi(\cdot)\text{ is decreasing on (0,1) and }\int_0^1\phi(z)dz\le 1 \right\},\\
  \unitdensity((0,1)):&=\left\{ \phi:(0,1)\mapsto\mathbb{R}_+ \mid \phi(\cdot)\text{ is decreasing on (0,1) and }\int_0^1\phi(z)dz= 1 \right\},
\end{align}
both of which can be identified as subsets of $\cM((0,1))$. Finally, for any random variables $X$, let $X\wedge 0:=\min(X,0)$.

\subsection{Definition}

A risk measure is a mapping which associates to a random variable
$X$, representing the future P\&L of a portfolio, a number $\rho(X)$
representing its risk. The set of random variables $X$ is often
taken to be $\LP^\infty(\Omega, \cF,\prob)$
\cite{artzner1999cmr,FollmerHSchiedA:02convexriskmeasure}, but one
may also allow for unbounded P\&Ls by defining  risk measures as
maps on $\LP^p(\Omega, \cF,\prob)$; see for instance,
\cite{FilipovicSvindland2008:CanonicalModelSpace} and the references
therein. In this section we will focus on risk measures defined on
$\LP^\infty(\Omega, \cF,\prob)$, following \cite{artzner1999cmr}. In
the study of robustness in Sections \ref{sec:robustness} and
\ref{sec:sensitivity}, we will revert back to the more general case
of unbounded losses.

\begin{definition}[Loss-based risk measures]\label{de:riskmeasure}
A mapping $\rho:\LP^\infty(\Omega, \cF,\prob)\rightarrow \mathbb{R}_{+}$ is called a {\em loss-based risk measure} if it satisfies
  \begin{itemize}
    \item[\bf (a)] {\em Normalization for cash losses:} for any $\alpha \in\mathbb{R}_{+}$, $\rho(-\alpha) = \alpha$;
    \item[\bf (b)] {\em Monotonicity:} for any $X,Y \in \LP^\infty(\Omega, \cF,\prob)$, if $ X \leq Y$, then $\rho(X) \geq
     \rho(Y)$; and
    \item[\bf (c)] {\em Loss-dependence:} for any $X\in \LP^\infty(\Omega, \cF,\prob)$, $ \rho(X) = \rho(X\wedge 0)$.
\end{itemize}
\end{definition}

The cash-loss property is a ``normalization" property stating that the risk of a (non-random) cash liability is its face value. This property ensures that the risk of a portfolio has the same unit as monetary payoffs. It is easy to observe that the cash-loss property is implied by the cash-additivity property in the set of axioms of coherent risk measures. However, the cash-additivity cannot be inferred from the cash-loss property. The monotonicity property says that a higher payoff leads to lower risk, which is a natural requirement for a meaningful risk measure and is also enforced in the definition of coherent risk measures. The loss-dependence property entails that the risk of a portfolio only depends on its losses, which is a desirable property when the risk measure is used to compute margin requirements or capital reserves. Loss-based risk measures can be seen as a special case of the more general notion of risk orders considered in \cite{DrapeauKupper2010:RiskPreferences}.

A loss-based risk measure $\rho$ is called {\em convex loss-based risk measure} if it satisfies
\begin{itemize}   \item[\bf (d)] {\em convexity:} for any $X,Y \in \LP^\infty(\Omega, \cF,\prob)$ and $0<\alpha<1$, $\rho(\alpha X+(1-\alpha)Y) \leq \alpha
    \rho(X)+(1-\alpha)\rho(Y)$.
\end{itemize}
As in the classical convex risk measures, the convexity property leads to a reduction of risk under diversification. Compared to convex risk measures, i.e., those risk measures satisfying the cash-additivity, monotonicity, and convexity properties, convex loss-based risk measures have an additional property---loss-dependence, and on the other hand, they replace the cash-additivity property with a weaker one---cash-loss property.

The following useful lemma shows that the monotonicity and convexity
properties together imply $\LP^\infty(\Omega,\cF,\prob)$-continuity
for a risk measure. A similar result was proved in \cite[Proposition
3.1]{RuszczynskiShapiro2006:OptimizationConvexRiskFunctions}. We
  provide a simpler proof, using elementary results.
\begin{lemma}\label{le:continuity} Any mapping $\rho:\LP^\infty(\Omega,\cF,\prob)\to\mathbb{R}$ that is monotone and convex is continuous on $\LP^\infty(\Omega,\cF,\prob)$.\end{lemma}
\begin{proof}
Consider a   sequence  $(X_n)_{n\geq 1}$ that converges to
$X$ in $\LP^\infty(\Omega,\cF,\prob)$. Then the sequence
$\varepsilon_n:=||X_n-X||_\infty$, $n\ge 1$, converges to zero
and we can assume without loss of generality that $\varepsilon_n\le
1,n\ge 1$. Let $\alpha_n:=\sqrt{\varepsilon_n}$, $\beta_n :=
\frac{\varepsilon_n}{\alpha_n} + ||X||_\infty + 1$, then
\begin{align*}
  (1-\alpha_n)X_n - \alpha_n\beta_n = X_n - \varepsilon_n-\alpha_n(||X||_\infty + 1+X_n)\le X
\end{align*}
where the inequality is because $X_n-X\le ||X_n-X||_\infty$ and
$||X||_\infty + 1+X_n\ge 0$. By the monotonicity and convexity of
$\rho$, we derive
\begin{align*}
  \rho(X)&\le \rho((1-\alpha_n)X_n-\alpha_n\beta_n)\\
  &\le (1-\alpha_n)\rho(X_n) + \alpha_n\rho(-\beta_n)\\
  &\le (1-\alpha_n)\rho(X_n) + \alpha_n\rho(-(||X||_\infty+2)),
\end{align*}
which leads to
\begin{align*}
  \rho(X_n)\ge \frac{\rho(X) - \alpha_n\rho(-(||X||_\infty+2))}{1-\alpha_n}.
\end{align*}
Letting $n\to\infty$ and using the fact that $\alpha_n$ converges to zero, we
immediately derive $\liminf_{n\rightarrow \infty}\rho(X_n)\ge
\rho(X)$, i.e., $\rho$ is lower semi-continuous in
$\LP^\infty(\Omega,\cF,\prob)$. A similar argument leads to the
upper semi-continuity of $\rho$ in $\LP^\infty(\Omega,\cF,\prob)$.
\end{proof}

Convex loss-based risk measures are special cases of cash-subadditive risk measures defined in
\cite{ElKarouiRavanelli:2009CashSubadditive}. Indeed, consider any $X\in \LP^\infty(\Omega,\cF,\prob)$
and $\alpha\ge 0$. For any $\epsilon\in(0,1)$, we have
\begin{align*}
  \rho((1-\epsilon)X - \alpha) &= \rho((1-\epsilon)X + \epsilon(-\frac{\alpha}{\epsilon}))\\
  &\le (1-\epsilon)\rho(X) + \epsilon\rho(-\frac{\alpha}{\epsilon}))\\
  & = (1-\epsilon)\rho(X) + \alpha,
\end{align*}
where the inequality is due to the convexity property and the
last equality is due to the cash-loss property. Letting
$\epsilon\downarrow 0$ and using the semi-continuity obtained in Lemma \ref{le:continuity}, we have
\begin{align*}
  \rho(X-\alpha)\le \rho(X) + \alpha,\quad \forall X\in \LP^\infty(\Omega,\cF,\prob),\alpha\ge 0,
\end{align*}
which also implies
\begin{align*}
  \rho(X+\alpha)\ge \rho(X) - \alpha,\quad \forall X\in \LP^\infty(\Omega,\cF,\prob),\alpha\ge 0.
\end{align*}
Thus, convex loss-based risk measures are cash-subadditive.

Although convex loss-based risk measures are special cases of cash-subadditive risk measures, they deserve to be highlighted and treated separately because the loss-dependence property is a desirable property in some of the risk management practices and thus needs to be enforced in the definition of risk measures. In addition, the cash-loss property, which is not implied by  cash-subadditivity, is also a reasonable property to define risk measures. Therefore, it is of great interest to carefully investigate this particular class of cash-subadditive risk measures---convex loss-based risk measures. Even further, in Sections \ref{sec:robustness} and \ref{sec:sensitivity}, we study the robustness of the risk estimators associated with a family of loss-based risk measures, not necessarily convex and thus not necessarily cash-subadditive.

\subsection{Representation theorem for convex loss-based risk measures}
We now give a characterization of convex loss-based risk measures that
have a semi-continuity property, the Fatou property. A (loss-based)
risk measure $\rho$ satisfies the {\em Fatou property} if, for any
 sequence $(X_n)_{n\geq 1}$ uniformly bounded in $\LP^\infty(\Omega,
\cF,\prob)$ such that $X_n \rightarrow X$ almost surely, we have $\rho(X) \leq \liminf_{n\rightarrow \infty} \rho(X_n)$.

\begin{theorem}\label{th:representationgeneral}
  The following are equivalent
  \begin{enumerate}
    \item $\rho$ is a convex loss-based risk measure satisfying the Fatou property.
    \item There exists a convex function $V :\cM(\Omega,\cF,\bP)\rightarrow [0,\infty]$ satisfying
    \begin{align}\label{eq:representation_conditionV}
      \inf_{\| Y \|_{1} \geq 1-\epsilon} V(Y) = 0 \quad \text{ for any } \quad \epsilon \in (0,1),
    \end{align}
     such that
     \begin{align}\label{eq:representation}
       \rho(X) =  - \inf_{Y \in \cM(\Omega,\cF,\bP)} \{\expect[(X\wedge 0)Y] + V(Y)\},\quad \forall X\in \LP^\infty(\Omega, \cF,\prob).
     \end{align}
  \end{enumerate}
\end{theorem}
\begin{proof}
 Assume that $\rho$ is a convex loss-based risk measure satisfying the Fatou property. Following Theorem 2.1 in \cite{DelbaenFSchachermayerW:94na} or Theorem 3.2 in \cite{delbaen2002coherent}, $\rho$ is lower-semi-continuous under weak* topology if and only if it satisfies the Fatou property. By  \cite[Theorem 6]{frittelli2002por} there exists
  $V :\LP^1(\Omega,\cF,\bP)\rightarrow (-\infty,\infty]$ such that
  \begin{align*}
       \rho(X) &=  - \inf_{Y \in \LP^1(\Omega,\cF,\bP)} \{\expect[XY] + V(Y)\}\\
       &=- \inf_{Y \in \LP^1(\Omega,\cF,\bP)} \{\expect[(X\wedge 0)Y] + V(Y)\},\quad \forall X\in \LP^\infty(\Omega, \cF,\prob),
  \end{align*}
  where the second equality is due to the loss-dependence property. Furthermore, we have the dual relation
  \begin{align*}
    V(Y) = \sup_{X\in \LP^\infty(\Omega, \cF,\prob)} \{-\rho(X) - \expect[XY]\},\quad \forall Y\in \LP^1(\Omega, \cF,\prob).
  \end{align*}

For any $Y\in \LP^1(\Omega, \cF,\prob)$, let $A:=\{Y<0\}$. If $\prob(A)>0$,
  \begin{align*}
    V(Y) &\ge \sup_{n\ge 1} \{-\rho(n\mathbf{1}_A) - \expect[n\mathbf{1}_AY]\} = \sup_{n\ge 1} \{- \expect[n\mathbf{1}_AY]\} = +\infty,
  \end{align*}
    where the first equality holds because $n\mathbf{1}_A \geq 0$ and $\rho$ is monotone with $\rho(0)=0$. Now, if $Y\ge 0$ and $\|Y\|_{1}>1$, we have
  \begin{align*}
    V(Y) & \ge \sup_{\alpha\ge 0} \{-\rho(-\alpha) - \expect[(-\alpha)Y]\}= \sup_{\alpha\ge 0} \{(\expect[Y]-1)\alpha\}=+\infty.
  \end{align*}
  Thus, the domain of $V$ lies in $\cM(\Omega,\cF,\bP)$. Next, it is easy to see that $V(Y)\ge 0$ for any $Y\in \LP^1(\Omega, \cF,\prob)$ because $\rho(0)=0$. Finally, for any $\epsilon \in(0,1)$ and any $\alpha>0$,
  \begin{align*}
    \alpha &= \rho(-\alpha) =  \sup_{Y \in \cM(\Omega,\cF,\bP)} \{\alpha\expect[ Y] - V(Y)\}\\
    & =  \max\left(\sup_{Y \in \cM(\Omega,\cF,\bP),\|Y\|_1< 1-\epsilon} \{\alpha\expect[Y] - V(Y)\},\sup_{Y \in \cM(\Omega,\cF,\bP),\|Y\|_1\ge 1-\epsilon} \{\alpha\expect[Y] - V(Y)\}\right)\\
    &\le \max\left(\alpha(1-\epsilon),\sup_{Y \in \cM(\Omega,\cF,\bP),\|Y\|_1\ge 1-\epsilon} \{\alpha\expect[Y] - V(Y)\}\right)\\
    &\le \max\left( \alpha(1-\epsilon), \alpha -\inf_{Y \in \cM(\Omega,\cF,\bP),\|Y\|_1\ge 1-\epsilon}\{ V(Y)\}\right).
  \end{align*}
  Thus, we conclude that $V(\cdot)$ must satisfy \eqref{eq:representation_conditionV}.

  On the other hand, one can check that $\rho$ represented in \eqref{eq:representation} is a convex loss-based risk measure satisfying the lower-semi-continuity under weak* topology and thus the Fatou property.
\end{proof}

We can see in the representation theorem that the domain of the
penalty function $V(\cdot)$ is a subset of $\cM(\Omega,\cF,\bP)$,
the set of all positive measures with total mass less than one, which contrasts with the representation theorem for convex risk measures in which the domain of the penalty function is a subset of $\cP(\Omega,\cF,\bP)$. This difference
is also observed in \cite[Theorem 4.3-(b)]{ElKarouiRavanelli:2009CashSubadditive}. Compared to the representation theorem in
\cite{ElKarouiRavanelli:2009CashSubadditive}, we have an additional
condition \eqref{eq:representation_conditionV}, which is due to the cash-loss
property. Moreover, the dual representation formula
\eqref{eq:representation} depends only on the negative part of $X$
due to the loss-dependence property.

If we compare the results in Theorem \ref{th:representationgeneral} with the representation of the general notion of risk orders in Theorem 2.6 of \cite{DrapeauKupper2010:RiskPreferences}, then we notice that Equation  \eqref{eq:representation} represents a particular set of monotone quasi-convex risk functionals $R$ in \cite{DrapeauKupper2010:RiskPreferences} due to the additional features of convex loss-based risk measures such as the convexity and cash-loss properties.

\subsection{Statistical loss-based risk measures}

Most of the risk measures used in finance are statistical, or distribution-based risk measures, i.e. they depend on $X$ only through its distribution $F_X(\cdot)$:
$$ F_X(\cdot) = F_Y(\cdot) \quad \Rightarrow  \quad \rho(X)=\rho(Y).$$

Following ideas from \cite{kusuoka2001lic},
\cite{FrittelliMEmanuelaRG:05lawinvariantriskmeasure}, \cite{Ruschendorf2006:LawInvariantConvexRiskMeasures},
\cite{jouini2006law}, and \cite{follmer2011sfi}, we derive a representation theorem for statistical convex
loss-based risk measures.
\begin{theorem}\label{th:representationlawinvariant}
  Let $\rho$ be a statistical convex loss-based risk measure.
  There exists a convex function $v :\density((0,1))\rightarrow [0,\infty]$ satisfying
    \begin{align}\label{eq:penaltyfunctionlawinvariant}
      \inf_{\int_0^1\phi(z)dz\ge 1-\epsilon} v(\phi) = 0\text{ for any }\epsilon \in (0,1),
    \end{align}
     such that
     \begin{align}\label{eq:representationlawinvariant}
       \rho(X) =  - \inf_{\phi \in \density} \left\{\int_0^1(G_X(z)\wedge 0)\phi(z)dz + v(\phi)\right\},\quad \forall X\in \LP^\infty(\Omega, \cF,\prob).
     \end{align}
\end{theorem}
\begin{proof}
First, we remark that a statistical convex loss-based risk measure necessarily
satisfies the Fatou property. Indeed, as noted in Theorem 2.2,
\cite{jouini2006law}, distribution-based convex functionals in
$\LP^\infty(\Omega,\cF,\prob)$ satisfy the Fatou property if and only if
they are lower-semi-continuous under the $\LP^\infty(\Omega,\cF,\prob)$
norm. By Lemma \ref{le:continuity}, any convex loss-based risk measure is continuous under the $\LP^\infty(\Omega,\cF,\prob)$ norm, so automatically satisfies the Fatou
property if it is distribution-based. We then need to build a
connection between \eqref{eq:representation} and
\eqref{eq:representationlawinvariant}
  when the risk measure is distribution-based, which can be done following the lines of  \cite[Theorem 2.1]{jouini2006law}.
\end{proof}

We can observe that in the representation theorem the domain of the penalty function $v(\cdot)$ is a subset of $\density$, which by definition is the set of positive measures on $(0,1)$ that have decreasing densities and have total mass less than or equal to one. By contrast, in the representation theorem for statistical convex risk measures, the domain of the penalty function is a subset of $\unitdensity$, which is a set of probability measures on $(0,1)$ with decreasing densities.

Motivated by the representation \eqref{eq:representationlawinvariant}, we sometimes abuse the notation by writing
\begin{align}
  \rho(G(\cdot)) = - \inf_{\phi \in \density} \left\{\int_0^1(G(z)\wedge 0)\phi(z)dz + v(\phi)\right\}
\end{align}
for any bounded quantile functions, if we are considering a statistical convex loss-based risk measure.

\subsection{Loss-based versione of convex risk measures}

For any convex risk measure $\widetilde \rho$, we can define a new risk measure $\rho$ by applying $\rho$ to the loss part of each portfolio's P\&L, i.e., $\rho(X):=\widetilde \rho(X\wedge 0)$ for any $X\in\LP^\infty(\Omega, \cF,\prob)$. It is easy to verify that $\rho$ is a convex loss-based risk measure. We call $\rho$ the loss-based version of $\widetilde \rho$.

In the following, we show that a convex loss-based risk measure $\rho$ is the loss-based version of some convex risk measure if and only if it satisfies
  \begin{itemize}
    \item[\bf (e)] {\em cash-loss additivity:} for any $X\in \LP^\infty(\Omega, \cF,\prob)$, $X\le 0$, and $\alpha\in \mathbb{R}_+$, $\rho(X-\alpha)=\rho(X)+\alpha$.
  \end{itemize}
The cash-loss additivity property says that for a portfolio that generates a pure loss, extracting certain amount of cash from the portfolio will increase its risk by the same amount.

On the one hand, if $\rho(X)= \widetilde \rho(X\wedge 0)$ for certain convex risk measure $\widetilde \rho$, then for any $X\in \LP^\infty(\Omega, \cF,\prob)$, $X\le 0$, and $\alpha\in \mathbb{R}_+$,
\begin{align*}
  \rho(X-\alpha) = \widetilde \rho(X-\alpha) = \widetilde \rho(X)+\alpha = \rho(X)+\alpha,
\end{align*}
where the second equality is due to the cash-additivity property of $\widetilde \rho$.

On the other hand, suppose a convex loss-based risk measure $\rho$ satisfies the cash-loss additivity property. Define
  \begin{align*}
    \widetilde \rho(X) =  \rho(X-\alpha_X)-\alpha_X,
  \end{align*}
  where $\alpha_X$ is any upper-bound of $X$. By the cash-loss additivity property for $\rho$, $\widetilde \rho$ is well-defined. Furthermore, it is easy to check that $\widetilde \rho$ is a convex risk measure and $\rho(X) = \rho(X\wedge 0) = \widetilde\rho(X\wedge 0)$.

  However, even though $\rho$ satisfies the cash-loss additivity property, it is still not cash additive. Indeed, in general, we only have
  \begin{align*}
    \rho(X+\alpha) = \widetilde\rho((X+\alpha)\wedge 0)\ge \widetilde\rho(X\wedge 0 + \alpha) = \widetilde\rho(X\wedge 0) - \alpha = \rho(X) - \alpha,
  \end{align*}
  for any $X\in \LP^\infty(\Omega, \cF,\prob)$ and $\alpha \in \mathbb{R}$.

A natural question is whether any convex loss-based risk measure satisfies the cash-loss additivity, i.e., whether it is the loss-based version of certain convex risk measure. The answer is no, which is illustrated by a nontrivial and meaningful example presented in the following section. As a result, convex loss-based risk measures are nontrivial extensions of convex risk measures.

\subsection{Examples}\label{sec:examples}
\begin{example}[Put option premium]
  \cite{jarrow02} argues that a natural measure of a firm's insolvency risk is the premium of a put option
  on the firms equity, which is given by the positive part of its net value $X$ (assets minus liabilities), i.e. $\expect^{\mathbb{Q}}[- \min(X,0)]$ where ${\mathbb{Q}}$ is an appropriately chosen
   pricing model.
  One can generalize this to any portfolio whose net value is represented by a random variable $X$:  the downside risk of the
  portfolio can be measured by
  \begin{align}\label{eq:putpremium}
    \rho(X):= \expect[- \min(X,0)]  \end{align}
 This example satisfies all the properties in Definition
 \ref{de:riskmeasure} and the convexity property: it is a convex loss-based risk measure. In particular, as noted by \cite{jarrow02}, it is not cash-additive. In the actuarial literature such risk measures have existed for more than 150 years, see \cite{Hattendorff1868}.
\end{example}

\begin{example}[Scenario-based margin requirements]
\iffalse
When determining margin requirements for derivative transactions,
the objective of a central  clearing facility is to compute the
margin requirement of each participant in the clearinghouse in order
to cover losses incurred by the clearing participants portfolio over
a liquidation horizon $T$ (typically, a few days). A popular method
for computing such margin requirements --used for example by various
futures and options exchanges-- is to select a certain number of
``stress scenarios" for risk factors affecting the portfolio and
compute the margin requirement as the maximum loss over these
scenarios. If one denotes the portfolio P\&L over the horizon $T$ by
$X$, then the margin requirement $\rho(X)$ is given by
  \begin{align}\label{eq:SPAN}
    \rho(X)= \max \{ - \min(X(\omega_1),0),..., - \min(X(\omega_n),0)\}  \end{align}
where $\omega_1,...,\omega_n$ are the stress scenarios. Naturally,
the clearinghouse only consider  the losses of the portfolio when
computing the margin: as a result, such margin requirements may be
viewed as a loss-based risk measure.
 This example satisfies all the properties of Definition
 \ref{de:riskmeasure}: it is a loss-based risk measure.
 This is the main idea behind the SPAN method used by the Chicago Merchantile
 Exchange (CME).
 Interestingly, the   SPAN method was considered as an
 initial motivation for the definition of coherent risk measures in
 \cite{artzner1999cmr}. Yet it is easy to check that
 \eqref{eq:SPAN} is loss-dependent and therefore not cash-additive,
 so is not a coherent risk measure.
 \fi
 It is easy to verify that the margin requirement \eqref{eq:SPAN} that is used in the CME is a convex loss-based risk measure. This method of determining margin requirements is known as the SPAN method. Interestingly, this method was considered as an
 initial motivation for the definition of coherent risk measures in
 \cite{artzner1999cmr}. Yet it is easy to check that the margin requirement \eqref{eq:SPAN} is not cash-additive, thus not a coherent risk measure.
\end{example}

\begin{example}[Expected tail-loss]
  This popular risk measure is defined as
  \begin{align}\label{eq:esriskmeasure}
    \rho(X):=-\frac{1}{\beta}\int_0^\beta(G_X(z)\wedge 0)dz,\quad X\in \LP^\infty(\Omega, \cF,\prob)
  \end{align}
  Note that, by construction, this risk measure focuses on the left tail of the loss distribution, since it only involves the quantile function on $(0,\beta)$. Nonetheless, the classical definition of expected shortfall $-\frac{1}{\beta}\int_0^\beta G_X(z) dz$ does not satisfy the loss-dependence property in Definition \ref{de:riskmeasure} since $G_X(\beta)$ might be greater than $0$ for some $X$. Therefore, we insert $G_X(z)\wedge 0$ in its definition to turn it into a loss-based risk measure. We notice that the put option premium is an expected tail-loss by taking $\beta =1$.
\end{example}

\begin{example}[Spectral loss measures]
A large class of statistical loss-based risk measures is obtained
by taking weighted averages of loss quantiles with various weight
functions  $\phi\in\unitdensity$:
  \begin{align}\label{eq:spectralriskmeasure}
    \rho(X):=-\int_0^1(G_X(z)\wedge 0)\phi(z)dz,\quad X\in \LP^\infty(\Omega, \cF,\prob)
  \end{align}
We call such a  risk measures a {\em  spectral loss measure}. By definition, this risk measure is the loss-based version of the spectral risk measures defined by \cite{acerbi2002smr}. As a result, it is a convex loss-based risk measure, and in addition it satisfies the positive homogeneity property:
for any $\lambda>0$, $\rho(\lambda X) = \lambda \rho(X)$. Notice
that the expected tail-loss is a spectral loss measure with
$\phi(z)=\frac{1}{\beta}1_{(0,\beta)}(z)$.
\end{example}

\begin{example}[Loss certainty equivalent]
  Consider $u(\cdot)\in C^4(\bR_+)$ which is strictly increasing and strictly convex. Assume $u'/u''$ is concave. Consider the following mapping:
  \begin{align}\label{eq:CEriskmeasure}
    \rho(X) := u^{-1}\left(\expect u(|X\wedge 0|)\right),\quad X\in \LP^\infty(\Omega, \cF,\prob).
  \end{align}
  It is clear that $\rho$ satisfies the cash-loss, monotonicity, and loss-dependence properties. Therefore, it is a loss-based risk measure. On the other hand,  by \cite[Theorem 106]{HardyLittlewoodPolya:1959Inequalities}, $\rho$ is also convex. Thus, $\rho$ is a convex loss-based risk measure. We call $\rho$ a {\em loss certainty equivalent }. By definition, $\rho$ is distributional-based and
  \begin{align}\label{eq:CEriskmeasurestatistical}
    \rho(X) = u^{-1}\left(\int_0^1u(|G_X(z)\wedge 0|)dz\right),\quad X\in \LP^\infty(\Omega, \cF,\prob).
  \end{align}

  If $u(x) = x^p,x\ge 0$ for certain $p\ge 1$, we will speak of the {\em $L^p$ loss certainty equivalent }.
  Here, when $p=1$, $u(\cdot)$ is not strictly convex. However, in this case, \eqref{eq:CEriskmeasure} is still well-defined and the risk measure is actually the put option premium, which is also a convex loss-based risk measure. Thus, we include the case $p=1$ here and identify the put option premium as a special case of $L^p$ loss certainty equivalent  in the following.
  One can show that the $L^p$ loss certainty equivalent is the only loss certainty equivalent that satisfies the positive homogeneity property. The $L^p$ loss certainty equivalent  has the following dual representation
  \begin{align}\label{eq:Lpriskmeasurerepresentation}
    \rho(X) = -\inf_{Y\in \cM^q(\Omega,\cF,\prob)}\expect\left[(X\wedge 0) Y\right]
  \end{align}
  where $1<q\le \infty$ is the conjugate of $p$, i.e., $\frac{1}{p}+\frac{1}{q}=1$, and $\cM^q(\Omega,\cF,\prob)$ is the set of all nonnegative random variables with $L^q$ norms less than or equal to one. Moreover, it has the distribution-based representation
  \begin{align}\label{eq:Lpriskmeasurerepresentationlawinvariant}
    \rho(X) = -\inf_{\phi\in \density^q((0,1))}\int_0^1(G_X(z)\wedge 0)\phi(z)dz.
  \end{align}
  where $\density^q((0,1))$ is the set of $\phi\in\density((0,1))$ such that $\int_0^1\phi(z)^qdz\le 1$.

  The $L^p$ certainty equivalents are closely related to the {\em lower partial moments} in \cite{Fishburn1977:MeanRiskAnalysisBelowTargetReturns}. It is straightforward to show that $L^p$ loss certainty equivalents for $p>1$ do not satisfy the cash-loss additivity property, so are loss-based versions of certain convex risk measures. Therefore, $L^p$ loss certainty equivalents, which are an important family of risk measures and cannot be accommodated in the framework of convex risk measures, are convex loss-based risk measures.

  Let $u(x) = e^{\beta x},x\ge 0$ for certain $\beta>0$, then the loss certainty equivalent becomes the {\em entropic loss certainty equivalent }, a loss-based version of the entropic risk measure studied in \cite{frittelli2002por}, \cite{FollmerHSchiedA:02convexriskmeasure}, and \cite{follmer2011sfi}. Actually, one can show that the entropic loss certainty equivalent and $L^1$ loss certainty equivalent are the only loss certainty equivalents that satisfy the cash-loss additivity property, see for instance Proposition 2.46 in \cite{follmer2011sfi}.

  By recalling the representation theorems for the entropic risk measure,\footnote{For the representation theorems for the entropic risk measure, see for instance, \cite{frittelli2002por,FrittelliMEmanuelaRG:05lawinvariantriskmeasure}, \cite{FollmerHSchiedA:02convexriskmeasure}, and \cite{follmer2011sfi}.} we obtain the dual representation for the entropic loss certainty equivalent
  \begin{align}\label{eq:entropyriskmeasurerepresentation}
    \rho(X) = -\inf_{Y\in \cP(\Omega,\cF,\prob)}\{\expect[(X\wedge 0)Y] + V(Y)\},
  \end{align}
  where
  \begin{align}
    V(Y) = \expect\left[Y\ln Y\right]- \inf_{Y\in \cP(\Omega,\cF,\prob)}\expect\left[Y\ln Y\right],
  \end{align}
  and the distributional-based representation
  \begin{align}\label{eq:entropyriskmeasurerepresentationlawinvariant}
    \rho(X) = -\inf_{\phi\in \unitdensity((0,1))}\int_0^1(G_X(z)\wedge 0)\phi(z)dz + v(\phi),
  \end{align}
  where
  \begin{align}
    v(\phi) = \int_0^1\phi(z)\ln \phi(z)dz-\inf_{\phi\in\unitdensity((0,1))} \int_0^1\phi(z)\ln \phi(z)dz.
  \end{align}
\end{example}

\section{Robustness of Risk Estimators}
\label{sec:robustness} In practice, for measuring the risk of a
portfolio, aside from the theoretical choice of a risk measure, a
key issue is the estimation of the risk measure, which requires the
choice of a {\it risk estimator} \citep{contdeguestscandolo}.
 In this
section we study the robustness property of empirical risk
estimators built from some statistical loss-based risk measure. We
follow the ideas in \citet{contdeguestscandolo} but we will view
risk measures as functionals on the set of quantile functions,
rather than the set of distribution functions. As we will see, this
makes the study of continuity properties of loss-based risk measures
easier. Moreover, from Theorem \ref{th:representationlawinvariant}, a statistical convex loss-based risk measure can be represented by \eqref{eq:representationlawinvariant}, which suggests that it is more natural to work directly on quantile functions. Note that the quantile representation of entropic risk measures and their different variants also appear in \cite{FollmerKnispel2012:ConvexCapitalRequirments}.

Denote by $\quan$ the set of all quantile functions and by $\cdf$ the set of all distribution functions. The L\'evy-Prokhorov metric between two distribution functions $F_1\in\cdf$ and $F_2\in\cdf$ is defined as
\begin{equation*}
   d_P(F_1,F_2)\triangleq\inf\{\epsilon>0
\,:\,F_1(x-\epsilon)-\epsilon\leq F_2(x)\leq F_1(x+\epsilon)+\epsilon,\;\forall x\in\bR\}.
\end{equation*}
This metric appears to be the most tractable one on $\cdf$ and it induces the same topology as the usual weak topology on $\cdf$.

The quantile set $\quan$ and the distribution set $\cdf$ are connected by the following one-on-one correspondence
\begin{align*}
  \begin{array}{ccc}
  \cdf &\rightarrow &\quan\\
  F(\cdot) & \mapsto & F^{-1}(\cdot)
  \end{array}
\end{align*}
where
\begin{align*}
  F^{-1}(t) = \inf\left\{x\in\mathbb{R}\mid F(x)\ge t\right\},\quad t\in (0,1)
\end{align*}
is the left continuous inverse of $F(\cdot)$. Such a correspondence, together with the L\'evy-Prokhorov metric on $\cdf$ induces a metric on $\quan$ which we denote by $d$. The convergence under this metric can be characterized by the following: for any $G_n,G\in \quan$, $G_n\rightarrow G$ if and only if $G_n(z)\rightarrow G(z)$ at any continuity points of $G$. In the following, we only need the characterization of the convergence on $\quan$, so the choice of metric on $\quan$ is irrelevant once it leads to the same topology.

Most of the time, we work with quantile functions that are continuous on $(0,1)$ in order to avoid irregularities due to the presence of atoms. In practice, it is not restrictive to focus on continuous quantile functions. Indeed, people do assume the continuity of quantile functions in many applications, e.g., when computing the VaR. The study of discontinuous quantile functions is more technical and of little interest, so we choose not to pursue in this direction. In the following, we denote by $\quan_c$ the set of all continuous quantile functions. We also denote by $\quan^\infty$ the set of all bounded quantile functions and by $\quan^\infty_c$ the set of all bounded continuous quantile functions.

\subsection{A family of statistical loss-based risk measures}
Motivated by the representation \eqref{eq:representationlawinvariant}, we consider a family of statistical loss-based risk measures defined by
the following {\em Fenchel-Legendre transform}:
\begin{align}
\label{eq:rhofunction}
  \rho(G(\cdot)):= -\inf_{m\in \dom(v)}\left\{\int_{(0,1)} (G(z) \wedge 0) m(dz) + v(m)\right\},
\end{align}
where $v:\cM((0,1))\rightarrow [0,\infty]$ is the {\em penalty function} satisfying
\begin{align}\label{eq:penaltyfunction}
      \inf_{m((0,1))\ge 1-\epsilon} v(m) = 0\text{ for any }\epsilon \in (0,1),
\end{align}
and $\dom(v)$ is the domain of $v$, i.e.,
\begin{align}
  \dom(v):=\{m\in \cM((0,1))\mid v(m)<\infty\}.
\end{align}
It is easy to see that $\rho(G(\cdot))$ is well-defined for any
$G(\cdot)\in \quan$, and $\dom(\rho)$, the domain of $\rho$,
contains $\quan^\infty$.

Straightforward calculation shows that the risk measure $\rho$ defined in \eqref{eq:rhofunction} satisfies the cash-loss, monotonicity, and loss-dependence properties, so it is a loss-based risk measure. It is clear to see from the representation \eqref{eq:representationlawinvariant} that statistical convex loss-based risk measures are special cases of the risk measures in \eqref{eq:rhofunction} by identifying each element in $\density$ as the density of a certain measure in $\cM((0,1))$. The risk measures \eqref{eq:rhofunction}, however, do not necessarily satisfy the convexity property. For instance, if we choose $\dom(v)=\{\delta_{\alpha}\}$ and $v(\delta_{\alpha})=0$, where $\delta_\alpha$ is the Dirac measure at $\alpha\in(0,1)$, then $\rho(G(\cdot)) = -G(\alpha)\wedge 0$, which is the $\alpha$-level VaR on losses. It is well-known that despite its wide use in practice, VaR does not satisfy the convexity property. To conclude, the risk measures defined by \eqref{eq:rhofunction} are a rich family that include both statistical convex loss-based risk measures that entail the convexity property and the VaR on losses that is popular in practice.

Finally, it is easy to observe that although $\rho$ in \eqref{eq:rhofunction} does not satisfy the convexity property, when viewed as a mapping on $\quan$, it satisfies the
following property.
\begin{enumerate}
  \item[{\bf (d')}] {\em Quantile convexity:} for any $G_1(\cdot),G_2(\cdot)\in \quan$ and $0<\alpha<1$, $\rho(\alpha G_1(\cdot)+(1-\alpha)G_2(\cdot))\leq \alpha \rho(G_1(\cdot))+(1-\alpha)\rho(G_2(\cdot))$.
\end{enumerate}
Quantile convexity is related to  co-monotonic subadditivity, discussed in \cite{heyde2006grm}, and co-monotonic convexity  \citep{SongYan:2009riskmeasure}.

\subsection{Qualitative robustness}
In practice, in order to compute the risk measure $\rho(G(\cdot))$ of a portfolio whose P\&L has quantile $G(\cdot)$, one has first to estimate the distribution or quantile of the portfolio's P\&L from data, and then apply the risk measure $\rho$ to the estimated distribution or quantile. One of the popular ways is to apply the risk measure to the empirical distribution.

A {\em sequence of samples} of $G(\cdot)\in \quan$ is a sequence of random variables $X_1,X_2,\cdots$ which are i.i.d and follow the distribution $G^{-1}$. We denote by $\mathbf{X}$ this sequence of samples and $\mathbf{X}^n$ its first $n$ samples. For each sample size $n$, the {\em empirical distribution} is defined as
\begin{align}
  F^{\text{emp}}_{\mathbf{X}^n}(x) = \frac{1}{n}\sum_{i=1}^n\mathbf{1}_{X_i\le x},\quad x\in\mathbb{R},
\end{align}
and the {\em empirical quantile} is defined as
\begin{align}
  G^{\text{emp}}_{\mathbf{X}^n}(z) :=(F^{\text{emp}}_{\mathbf{X}^n})^{-1}(z) = X_{(\lfloor nz\rfloor +1)},\quad z\in (0,1),
\end{align}
where $\lfloor a\rfloor$ denotes the integer part of $a$ and $X_{(1)} \leq \cdots \leq X_{(n)}$. In practice, the quantity $\re(\mathbf{X}^n) := \rho(G^{\text{emp}}_{\mathbf{X}^n}(\cdot))$ is computed as the estimated risk measure $\rho(G(\cdot))$ given the $n$ samples $\mathbf{X}^n$. The risk estimator $\re$ is defined on $\cup_{n\geq1} \bR^n$, the set of all possible sequences of samples $(\mathbf{X}^n)_{n\geq1}$, and has values in $\bR^+$. Since the samples can be regarded as random variables, so is the risk estimator $\re(\mathbf{X}^n)$. We denote by $\cL_n(\re,G)$ the distribution function of $\re(\mathbf{X}^n)$.

$\rho$ is said to be {\em consistent} at $G=F^{-1}\in\dom(\rho)$ if
\begin{align}
  \lim_{n\rightarrow \infty} \re(\mathbf{X}^n) = \rho(F^{\text{emp}}_{\mathbf{X}^n})^{-1}) \qquad \text{almost surely.}
\end{align}
Because the true risk measure $\rho(G(\cdot))$ is estimated by $\re(\mathbf{X}^n)$, the consistency is the minimal requirement for a meaningful risk measure. In the following we denote by $\quan^\rho$ the set of quantiles $G$ at which $\rho$ is well defined and consistent, and $\quan^\rho_c$ the continuous quantiles in $\quan^\rho$.

The following definition of {\it robust risk estimator} is
considered in \cite[Definition 4]{contdeguestscandolo}:
\begin{definition}[\citealt{contdeguestscandolo}]
Let $\rho$ be defined by \eqref{eq:rhofunction} and
$\cC\subset\quan^\rho$ be a  a set of plausible P\&L
quantiles. $\re$ is {\em $\cC$-robust} at $G\in \cC$ if for any
$\varepsilon>0$ there exist $\delta>0$ and $n_0\ge 1$ such that, for
all $\widetilde G\in \cC$,
\begin{align*}
d(\widetilde G,G)\le \delta \Longrightarrow d_P(\cL_n(\re,\widetilde G),\cL_n(\re,G))<\varepsilon, \qquad \forall n \geq n_0.
\end{align*}
A risk estimator $\re$ is called {\em $\cC$-robust} if it is
$\cC$-robust at any $G\in\cC$.
\end{definition}
We can see that the definition does not rely on the choice of the
metric on the topological space $\quan$. The choice of the
L\'evy-Prokhorov distance on $\cdf$, however, is critical. As
pointed out by \cite{huber1981rs}, the use of a different metric,
even if it also metrizes the weak topology, may lead to a different
class of robust estimators. This metric   is a natural choice in
robust statistics \citep{huber1981rs}. Alternative formulations are
proposed by \cite{kratschmer2011qualitative}.

The following proposition, taken from \cite{contdeguestscandolo},
shows that the robustness of the risk estimator $\re$ is equivalent
to the continuity of the risk measure $\rho$ on $\quan$ under the
weak topology.
\begin{proposition}[\citealt{contdeguestscandolo}]\label{prop:robustness}
  Let $\rho$ be a risk measure and $G\in \cC\subset \quan^\rho$. The following are equivalent:
  \begin{enumerate}
    \item $\rho$, when restricted to $\cC$, is continuous at $G$;
    \item $\re$ is $\cC$-robust at $G$.
  \end{enumerate}
\end{proposition}

In the following, we are going to investigate the continuity of $\rho$, which finally clarifies whether $\re$ is robust or not. The following lemma is useful.
\begin{lemma}\label{le:uniformconvergence}
  Let $-\infty<a<b<+\infty$, and $G_n,G$ be increasing functions on $[a,b]$. Suppose $G(z)$ is continuous on $[a,b]$ and $G_n(z)\rightarrow G(z)$ for each $z\in [a,b]$, then $G_n(z)\rightarrow G(z)$ uniformly on $[a,b]$.
\end{lemma}
\begin{proof}
    Because $[a,b]$ is a closed interval and $G(z)$ is continuous on $[a,b]$, for each $\epsilon>0$, there exists $a=z_1<\cdots<z_m=b$ such that $\sup_{1\le i\le m-1}|G(z_{i+1})-G(z_i)|<\frac{\epsilon}{2}$. On the other hand, there exists $N$ such that when $n\ge N$, $|G_n(z_i)-G(z_i)|<\frac{\epsilon}{2}$ for every $z_i,i=1,\dots m$. Now, for any $z\in [a,b]$, there exist $z_i,z_{i+1}$ such that $z_i\le z\le z_{i+1}$. Thus,
  \begin{align*}
    G_n(z)-G(z)&\le G_n(z_{i+1})-G(z_i)\\
    & = G_n(z_{i+1})-G(z_{i+1})+G(z_{i+1})-G(z_i)\\
    &< \epsilon
  \end{align*}
  when $n\ge N$. Similarly, we have $G_n(z)-G(z)>-\epsilon$ when $n\ge N$. Therefore, $G_n(z)\rightarrow G(z)$ on $[a,b]$ uniformly.
\end{proof}

The following lemma shows that any risk measure $\rho$ in \eqref{eq:rhofunction} is consistent on $\quan_c^\infty$, i.e., $\quan_c^\infty\subset \quan^\rho_c$.
\begin{lemma}\label{le:consistency}
  Let $\rho$ be given in \eqref{eq:rhofunction}. Then $\rho$ is consistent at any $G\in \quan_c$ that is bounded from below. In particular, $\quan_c^\infty \subset\quan^\rho_c$.
\end{lemma}
\begin{proof}
    Let $G\in \quan_c$ that is bounded from below, and $X_1,X_2,\dots$ be its samples. By Glivenko-Cantelli theorem, $G^{\text{emp}}_{\mathbf{X}^n}(z)\rightarrow G(z),0<z<1$ almost surely. Furthermore, $\inf_{i=1,\dots,n}X_i\rightarrow \text{essinf} X_1$ almost surely, which shows that $G^{\text{emp}}_{\mathbf{X}^n}(0+)\rightarrow G(0+)$. Thus, if we extend $G^{\text{emp}}_{\mathbf{X}^n}$ and $G$ from $(0,1)$ to $[0,1)$ by setting $G^{\text{emp}}_{\mathbf{X}^n}(0):=G^{\text{emp}}_{\mathbf{X}^n}(0+)$ and $G(0):= G(0+)$, then $G(\cdot)$ is continuous on $[0,1)$ and $G^{\text{emp}}_{\mathbf{X}^n}(z)\rightarrow G(z),0\le z<1$ almost surely.

    In the following, for each fixed $\omega$, let $G_n:=G^{\text{emp}}_{\mathbf{X}^n}$. For simplicity, we work with $U(\cdot) := -\rho(\cdot)$. We want to show that $U(G_n(\cdot))\rightarrow U(G(\cdot))$.  On the one hand,
    \begin{align*}
    \limsup_{n\rightarrow \infty}U(G_n(\cdot))&\le \inf_{m\in \dom(v)}\left[\limsup_{n\rightarrow \infty}\int_{(0,1)}(G_n\wedge0)(z)m(dz) + v(m)\right]\\
    &\le\inf_{m\in \dom(v)}\left[\int_{(0,1)}(G(z)\wedge0)m(dz) + v(m)\right]\\
    & = U(G(\cdot)),
    \end{align*}
  where the second inequality is due to Fatou's lemma. On the other hand, for each $\eta<1$, by Lemma \ref{le:uniformconvergence}, $G_n(z)\rightarrow G(z)$ uniformly for $z\in(0,\eta]$. Thus, we have
  \begin{align*}
    \liminf_{n\rightarrow \infty}U(G_n(\cdot))& =  \liminf_{n\rightarrow \infty}\inf_{m\in \dom(v)}\left[\int_{(0,1)}(G_n(z)\wedge 0)m(dz) + v(m)\right]\\
    & \geq \liminf_{n\rightarrow \infty}\inf_{m\in \dom(v)}\left[\int_{(0,1)}(G_n(z\wedge \eta)\wedge 0)m(dz) + v(m)\right]\\
     & = \inf_{m\in \dom(v)}\left[\int_{(0,1)}(G(z\wedge \eta)\wedge 0)m(dz) + v(m)\right]\\
     & = U(G(\cdot\wedge\eta)),
  \end{align*}
  where the second equality holds because $G_n(z\wedge \eta)\rightarrow G(z\wedge \eta)$ for $z\in(0,1)$ uniformly. Finally,
  \begin{align*}
    U(G(\cdot))\ge U(G(\cdot\wedge\eta)) &= \inf_{m\in \dom(v)}\left[\int_{(0,1)}(G(z\wedge \eta )\wedge 0)m(dz) + v(m)\right]\\
    &=\inf_{m\in \dom(v)}\Big[\int_{(0,1)}(G(z )\wedge 0)m(dz)+ v(m)\\
    &\qquad  -\int_{(0,1)} (G(z)\wedge 0-(G(z\wedge \eta )\wedge 0))m(dz)\Big]\\
    &\ge U(G(\cdot)) - \left[\lim_{z\uparrow 1}G(z)\wedge 0 - G(\eta)\wedge 0\right]\\
    &\rightarrow U(G(\cdot))
  \end{align*}
  as $\eta \uparrow 1$. Therefore, we conclude $\liminf_{n\rightarrow \infty}U(G_n(\cdot))\ge U(G(\cdot))$.
\end{proof}

Lemma \ref{le:consistency} shows that any risk measure $\rho$ defined in \eqref{eq:rhofunction} is consistent at least at bounded continuous quantile functions. For any particular example of the risk measures in \eqref{eq:rhofunction}, it is possible to show that it is consistent at certain unbounded quantile functions.\footnote{For instance, \citet[Example 2.11]{contdeguestscandolo} show that spectral risk measures are consistent at any quantile functions at which the risk measures are well-defined.} In this paper, we consider the general risk measure in \eqref{eq:rhofunction} and mainly focus on investigating the robustness of risk estimators. For this reason, we do not explore the consistency for any particular risk measure in detail.

The following result provides a sufficient and necessary condition under which the risk measures defined in \eqref{eq:rhofunction} are
continuous on some subset of $\quan$.
\begin{theorem}
\label{th:weak_continuity}
Let $\rho$ be defined by \eqref{eq:rhofunction} and $\cC$ be any subset of $\quan_c$ such that $\cC\supseteq\quan_c^\infty$. The following are equivalent:
\begin{itemize}
\item[(i)] $\rho(\cdot)$, when restricted to $\cC$, is continuous at any $G(\cdot) \in \cC$.
\item[(ii)] There exists $0<\delta<1$ such that
\ba
\sup_{m\in\dom(v)} m((0,\delta))=0.
\ea
\end{itemize}
Furthermore, if (ii) holds, $\quan_c^\rho = \quan_c$.
\end{theorem}

\begin{proof}
\begin{itemize}
\item[$(ii)\Rightarrow(i)$] The proof is analogous to the proof of Lemma \ref{le:consistency}.
  Let (ii) hold for some $0<\delta<1$. For simplicity, we work with $U(.) \triangleq -\rho(.)$ and first show that $U$ is lower-semi-continuous. For each $\eta<1$, we have
  \begin{align*}
    \liminf_{n\rightarrow \infty}U(G_n(\cdot))& =  \liminf_{n\rightarrow \infty}\inf_{m\in \dom(v)}\left[\int_{[\delta,1)}(G_n(z)\wedge 0)m(dz) + v(m)\right]\\
    & \geq \liminf_{n\rightarrow \infty}\inf_{m\in \dom(v)}\left[\int_{[\delta,1)}(G_n(z\wedge \eta)\wedge 0)m(dz) + v(m)\right]\\
     & = \inf_{m\in \dom(v)}\left[\int_{[\delta,1)}(G(z\wedge \eta)\wedge 0)m(dz) + v(m)\right]\\
     & = U(G(\cdot\wedge\eta)),
  \end{align*}
  where the second equality holds because $G_n(z\wedge \eta)$ converges to $G(z\wedge\eta)$ uniformly for $z\in[\delta,1)$ where Lemma \ref{le:uniformconvergence} applies. Then, by monotonicity, we have $0\le U(G(\cdot))-U(G(\cdot\wedge \eta))$. Finally,
  \begin{align*}
    U(G(\cdot))\ge U(G(\cdot\wedge\eta)) &= \inf_{m\in \dom(v)}\left[\int_{[\delta ,1)}(G(z\wedge \eta )\wedge 0)m(dz) + v(m)\right]\\
    &=\inf_{m\in \dom(v)}\Big[\int_{[\delta,1)}(G(z )\wedge 0)m(dz)+ v(m)\\
    &\qquad  -\int_{(0,1)} (G(z)\wedge 0-(G(z\wedge \eta )\wedge 0))m(dz)\Big]\\
    &\le U(G(\cdot)) - \left[\lim_{z\uparrow 1}G(z)\wedge 0 - G(\eta)\wedge 0\right]\\
    &\rightarrow U(G(\cdot))
  \end{align*}
  as $\eta \uparrow 1$.
  Thus, we have $ \liminf_{n\rightarrow \infty}U(G_n(\cdot))\ge U(G(\cdot))$, and can conclude that $U$ is lower-semi-continuous.

Next, we show that $U$ is also upper-semi-continuous.
  \begin{align*}
    \limsup_{n\rightarrow \infty}U(G_n(\cdot))&\le \inf_{m\in \dom(v)}\left[\limsup_{n\rightarrow \infty}\int_{(0,1)}(G_n\wedge0)(z)m(dz) + v(m)\right]\\
    &\le\inf_{m\in \dom(v)}\left[\int_{(0,1)}(G(z)\wedge0)m(dz) + v(m)\right]\\
    & = U(G(\cdot)),
  \end{align*}
  where the second inequality is due to Fatou's lemma. Together with the lower-semi-continuity, $U(\cdot)$ is continuous at $G$, and so is $\rho$.

  \item[$(i)\Rightarrow(ii)$] We prove it by contradiction. If (ii) is not true, there exists $\delta_n>0$, and $m_n\in\dom(v)$, such that $\delta_n\rightarrow 0$ and $m_n((0,\delta_n))>0$. Define
  \begin{align*}
    G_n(z):=\begin{cases}
      -\beta_n & 0<z\le\delta_n,\\
      \frac{\beta_n}{\delta_n}(z-2\delta_n)& \delta_n< z\le 2\delta_n\\
      0 & 2\delta_n<z<1,
    \end{cases}
    \quad n\ge 1,
  \end{align*}
  where
  \begin{align*}
    \beta_n:=\frac{v(m_n)+1}{m_n((0,\delta_n))},\quad n\ge 1.
  \end{align*}
  It is obvious that $G_n\in\quan_c^\infty\subset \cC$, $n\ge 1$.
  Because $\delta_n\downarrow 0$, $G_n(\cdot)\rightarrow  0$ in $\cC$. On the other hand,
  \begin{align*}
    U(G_n(\cdot))\le \int_{(0,1)}G_n(z)m_n(dz) + v(m_n)\le -1,
  \end{align*}
  and $U(0)=0$. Thus, $U(\cdot)$ is not continuous at $0$ which is a contradiction.
  \end{itemize}
Finally, if (ii) holds, for any $G(\cdot)\in \quan_c$, we have $\rho(G(\cdot)) = \rho(G(\cdot\vee \delta))$. Because $G(\cdot\vee \delta)$ is bounded from below, by Lemma \ref{le:consistency}, $\rho$ is consistent at $G(\cdot\vee \delta)$, and therefore consistent at $G(\cdot)$. In other words, $\quan_c^\rho = \quan_c$.
\end{proof}

Theorem \ref{th:weak_continuity} provides a sufficient and necessary condition for the risk measures in \eqref{eq:rhofunction} to be continuous under the weak topology on $\quan$. The reason for choosing the weak topology is because the resulting continuity result is directly connected to the robustness of the risk estimators according to Proposition \ref{prop:robustness}. In \cite{jouini2006law}, the continuity of statistical convex risk measures on $\quan$ under a different topology, named as {\em Lebesgue property}, is investigated. The authors find an equivalent characterization of the Lebesgue property. To compare the sufficient and necessary condition we derive in Theorem \ref{th:weak_continuity} and theirs, we extend the result in \cite{jouini2006law} to the case of the general risk measures in \eqref{eq:rhofunction} by mimicking the proof in that paper. Because the comparison is not the main theme of the paper, we place the details in Appendix \ref{se:lebesgue}.

Finally, we combine Proposition \ref{prop:robustness} and Theorem \ref{th:weak_continuity} to obtain a sufficient and necessary condition for the robustness of the risk estimators associated with the risk measures \eqref{eq:rhofunction}.
\begin{corollary}
\label{co:robustness}
Let $\rho$ be defined by \eqref{eq:rhofunction} and $\cC$ be any subset of $\quan_c^\rho$ such that $\cC\supseteq\quan_c^\infty$. Then the following are equivalent
\begin{enumerate}
\item $\re$ is $\cC$-robust
\item There exists $0<\delta<1$ such that
\ba
\sup_{m\in\dom(v)} m((0,\delta))=0.
\ea
\end{enumerate}
Furthermore, if $\re$ is $\cC$-robust, then $\quan_c^\rho = \quan_c$.
\end{corollary}

An immediate consequence of Corollary \ref{co:robustness} is that statistical convex loss-based risk measures do not lead to robust risk estimators.
\begin{corollary}
  Let $\rho$ be a loss-based statistical risk measure and $\cC$ be any subset of $\quan_c^\rho$ such that $\cC\supseteq \quan_c^\infty$. Then, $\re$ is not $\cC$-robust.
\end{corollary}
\begin{proof}
  By Theorem \ref{th:representationlawinvariant}, $\rho$ can be represented as \eqref{eq:representationlawinvariant} where the penalty function $v$ satisfies \eqref{eq:penaltyfunctionlawinvariant}. As a result, there exists a $\phi\in\Psi((0,1))\cap \dom (v)$ such that $\int_0^1\phi(z)dz\ge \frac{1}{2}$. Because $\phi(\cdot)$ is decreasing on (0,1), we must have $\int_0^\delta \phi(z)dz>0$ for any $\delta>0$. By Corollary \ref{co:robustness}, $\re$ is not $\cC$-robust.
\end{proof}

This result reveals a dilemma: one has choose between convexity property, which leads to a  reduction of risk under diversification, and the robustness of the risk estimator associated with this risk measure, which rendes it amenable to estimation and backtesting. In practice, when choosing a risk measure for a certain purpose, one has to decide which of the two properties is more important. For instance, if the risk measure is used to compute daily margin requirements in a clearing house, the robustness is the more important issue because a system generating unstable  margin requirements may lead to large margin calls even in absence of any significant market event. However, the convexity property might be more relevant if the risk measure is used as an allocation tool, rather than a risk management tool.

On the other hand, recalling that the VaR on losses can be identified as a special case of the risk measures in \eqref{eq:rhofunction} by letting $\dom(v)=\{\delta_{\alpha}\}$ and $v(\delta_{\alpha})=0$, Corollary \ref{co:robustness} shows that the 'histoical (loss-based) Value at Risk' i.e. the empirical loss quantile is a robust risk estimator.
 These results are similar in spirit to previous results by \cite{contdeguestscandolo}.

\subsection{Robustification of risk estimators}
\label{sec:robustification}
Corollary \ref{co:robustness} provides a sufficient and necessary condition for a loss-based risk measure $\rho$ represented by equation \eqref{eq:rhofunction} to be continuous and therefore to lead to a robust empirical risk estimator. In particular, convex loss-based statistical risk measures lead to non-robust risk estimators. In the following, we provide one way to robustify risk estimator computed from statistical convex loss-based risk measures. Fixing some $\delta\in(0,1)$, for any statistical convex loss-based risk measure $\rho$,\footnote{The discussion in this subsection can also be applied to the risk measures given by representation \eqref{eq:rhofunction}. The restriction to statistical convex loss-based risk measures is only made to illustrate further the conflict between the convexity property and the robustness.} consider its {\em $\delta$-truncation}
\ba
\label{eq:rhoc_phi_function}
\rhoc(G(\cdot)) := \rho(G(\cdot \vee \delta)),\quad G\in\quan.
\ea
The new risk measure $\rho_\delta$ leads to a new risk estimator, $\re_\delta$, by plugging historical quantile functions. In the following, we are going to show that $\re_\delta$ is robust. As a result, $\re_\delta$ can be regarded as robustification of $\re$, the risk estimator associated with $\rho$.

From the representation \eqref{eq:representationlawinvariant}, we can find the representation of $\rho_\delta$, which in turns shows that $\re_\delta$ is robust.
Define the map $\pi: \Psi \mapsto \cM((0,1))$ which associates to any density function $\phi$ in $\Psi$ a measure $m$ defined by
\baa
m(dz) :=
\left\{
\begin{array}{cc}
\phi(z)dz, & \delta<z<1, \\
\left(\int_{(0,\delta]} \phi(t)dt\right) \delta_{z}, & z=\delta, \\
0, & 0<z<\delta,
\end{array}
\right.
\eaa
where $\delta_{z}$ is the Dirac measure at $z$. The observation that $\pi$ is not a bijective map leads to an additional definition since the penalty function $v(m)$ for $m \in \pi(\Psi((0,1))):= \{m \mid m=\pi(\phi) \text{ for some } \phi \in \Psi((0,1)) \}$ cannot be derived uniquely from $v(\phi)$ where $m=\pi(\phi)$. Therefore, by denoting $\pi^{-1}(m) = \{\phi\in \Psi((0,1))\mid  \pi(\phi)=m \}$, we define for $m \in \pi(\Psi)$,
\baa
v_\delta(m) := \inf_{\phi  \in \Psi \cap \pi^{-1}(m)} v(\phi).
\eaa
From the representation \eqref{eq:representationlawinvariant}, we have
\ba
\rhoc(G) &=& -\inf_{\phi \in \Psi}\left\{(G(\delta) \wedge 0) \int_{(0,\delta]} \phi(z)dz + \int_{(\delta,1)} (G(z) \wedge 0) \phi(z)dz + v(\phi)\right\} \nonumber \\
&=& -\inf_{m \in \pi(\Psi)} \left\{ \int_{(0,1)} (G(z) \wedge 0) m(dz) + v_\delta(m)\right\}. \label{eq:rhocfunction}
\ea
Finally, from \eqref{eq:penaltyfunctionlawinvariant} it is easy to see that $v_\delta$ satisfies \eqref{eq:penaltyfunction}.

From the representation \eqref{eq:rhocfunction}, it is immediate to see that the $\delta$-truncation $\rho_\delta$ is no longer convex since measures $m \in \pi(\Psi((0,1)))$ have a point mass at $\delta$ and therefore do not admit a density on $(0,1)$. On the other hand, it is also straightforward to see that $\re_\delta$ is $\quan_c$-robust because each $m\in \pi(\Psi((0,1)))$ satisfies $m((0,\delta))=0$.

\begin{example}
  The $\delta$-truncation of  the spectral loss measure \eqref{eq:spectralriskmeasure} is given by
  \begin{align}
    \rho_\delta(G) = \int_\delta^1(G(z)\wedge 0)\phi(z)dz + G(\delta)\int_0^\delta \phi(z)dz,\quad G\in\quan.
  \end{align}
\end{example}

\begin{example}
  The $\delta$-truncation of the loss certainty equivalent  is given by
  \begin{align}\label{eq:CEriskmeasuretruncated}
    \rho_\delta(G(\cdot)) = u^{-1}\left(\int_0^1u(|G(t\vee \delta)\wedge 0|)dt\right),\quad G\in\quan.
  \end{align}
  In particular, the $\delta$-truncation of the $L^p$ loss certainty equivalent is given by
  \begin{align} \label{eq:Lpriskmeasuretruncated}
    \rho_\delta(G) = \left[\int_0^1|G(z\vee \delta)\wedge 0|^pdz\right]^{\frac{1}{p}},\quad G\in\quan,
  \end{align}
  and its representation is expressed as
  \begin{align}
    \rho_\delta(G) = -\inf_{m\in \pi(\density^q((0,1)))}\int_0^1(G(z)\wedge 0)m(dz),\quad G\in\quan.
  \end{align}
  The $\delta$-truncation of the entropic loss certainty equivalent is given by
\begin{align}\label{eq:entropyriskmeasuretruncated}
  \rho_\delta(G) = \frac{1}{\beta} \log \left(\int_{(0,1)} e^{-\beta \,G(z\vee \delta) \wedge 0}\, dz \right),\quad G\in \quan,
\end{align}
and its representation is expressed as
\ba
\rhoc(G) = -\inf_{m \in \pi(\Phi((0,1)))} \left\{ \int_{(0,1)} (G(z) \wedge 0) m(dz) + v_\delta(m)\right\},\quad G\in \quan
\ea
where for each $m\in\pi(\Phi((0,1)))$
\baa
v_\delta(m) &=& \inf_{\phi \in \Phi \cap \pi^{-1}(m)} v(\phi)\\
&=& \frac{1}{\beta} \, \int_{(\delta,1)} m'(z) \log\left(m'(z)\right) \, dz + \frac{1}{\beta} \, \lim_{z \downarrow \delta} \left( m'(z) \log\left(m'(z)\right) \, z\right) - \inf_{\phi\in\Phi((0,1))}\int_0^1\phi(z)\ln\phi(z)dz.
\eaa
  Here, $m'(z),\delta<z<1$ denotes the density of $m$ on $(\delta,1)$ w.r.t. to Lebesgue measure. This density is well-defined because $m \in\pi(\Phi((0,1)))$.
\end{example}

It is worth mentioning the paper \citet{contdeguestscandolo}, in which the authors also propose a way to robustify the risk estimator associated with the expected shortfall. The robustification suggested by those authors is different from ours. Indeed, \citet{contdeguestscandolo} propose to truncate the expected shortfall, denoted by $\textrm{ES}_\alpha$, in the following way
\baa
\text{ES}_{\delta,\alpha}(G) = \frac{1}{\alpha-\delta} \int_{(\delta,\alpha)} G(z) \, dz,
\eaa
and define the risk estimator associated with $\text{ES}_{\delta,\alpha}$ as the robustification. Applying their idea to the larger class of statistical convex loss-based risk measures would lead to defining the following truncation
\baa
\widetilde{\rhoc}(G) = -\inf_{\widetilde{\phi} \in \widetilde{\pi}(\Psi)} \, \left\{\int_{(0,1)} (G(z) \wedge 0) \widetilde{\phi}(z)dz + \widetilde{v_\delta}(\widetilde{\phi})\right\},
\eaa
where the map $\widetilde{\pi}: \Psi \rightarrow \cM((0,1))$ associates to any function $\phi \in \Psi$ another function $\widetilde{\phi} \in \cM((0,1))$ defined by
\baa
\widetilde{\phi}(z) := \frac{\phi(z) \mathbf{1}_{(\delta,1)}(z)}{\int_{(\delta,1)} \phi(z) dz} \not\in \Psi,
\eaa
and where the penalty function is given by
\baa
\widetilde{v}_\delta(\widetilde{\phi}) := \inf_{\phi \in \Psi \cap \widetilde{\pi}^{-1}(\widetilde{\phi})} v(\phi).
\eaa
Compared to the truncation $\rhoc$ considered in this paper, this new truncation $\widetilde{\rhoc}$ reassigns the probability weight attached to $z\in(0,\delta)$ evenly to $z\in(\delta,1)$. The new truncation $\widetilde{\rhoc}$ is less tractable than $\rhoc$ because it cannot be computed from the initial risk measure $\rho$ as in \eqref{eq:rhoc_phi_function}.

\section{Sensitivity Analysis of Risk Estimators}
\label{sec:sensitivity}
In the previous section, we have studied the robustness of risk estimators in a qualitative sense. One may argue that the above results rely on the choice of a topology (weak topology) together with a distance (L\'evy-Prokhorov distance) on the space of probability distributions. As illustrated  in Appendix \ref{se:lebesgue}, the choice of a weaker topology could lead to different continuity properties for the risk measures and thus to different robustness properties for the corresponding risk estimators. Nonetheless, as noted in \cite{huber1981rs}, the choice of the weak continuity to study robustness  is natural in statistics. To further illustrate this statement in this section, we study sensitivity properties of the risk estimators associated with loss certainty equivalents and their $\delta$-truncation versions without relying on any topology and show that the study leads to the same conclusions as before, i.e., loss certainty equivalents are not robust but their $\delta$-truncation versions are.

 The study of the sensitivity properties may be done by quantifying the sensitivity of risk estimators using  \textit{influence functions}   \citep{contdeguestscandolo,hampel1974influence}. Fix a risk estimator $\re$ for which the estimation is based on applying the risk measure $\rho$ to empirical quantile functions. Then, its sensitivity function at the quantile function $G$ of a distribution $F$, in the direction of the Dirac mass at $z$ is equal to
\begin{equation*}
\label{def:sensitivity}
    S(z;G) \triangleq \lim_{\epsilon\to 0^+}\frac{\rho(\epsilon \delta_z + (1-\epsilon)F)-\rho(F)}{\epsilon},
\end{equation*}
for any $z\in\bR$ such that the limit exists. Note that $S(z;G)$ is nothing but the directional derivative of the risk measure $\rho$ at $F$ in the direction $\delta_z\in\cD$.
$S(z,G)$ measures the sensitivity of the risk estimator based on a large sample to the addition of a new observation \cite{contdeguestscandolo}. In the former work, the authors consider different estimation methods, using both empirical and parametric distributions. In that case, the definition of the sensitivity function should be considered with more attention since the risk measure $\rho$ would have to be replaced with an \textit{effective risk measure} incorporating both the choice of the risk measure and the estimation method as explained in \cite{contdeguestscandolo}.

\subsection{Unbounded sensitivity functions}
In this section, we compute the sensitivity function of the loss
certainty equivalent. We find that this risk measure has unbounded
sensitivity function which is consistent with our findings of
Section \ref{sec:robustness}. Note that, unlike the setting of
Section \ref{sec:robustness}, this result makes no reference to any
topology on the set of loss distributions.
\begin{proposition} \label{prop:sensitivity_CE}
The sensitivity function of the loss certainty equivalent
\eqref{eq:CEriskmeasure} is given by
\begin{align}\label{eq:sensitivity_CE}
S(z;G) = \frac{u(|z\wedge 0|)-u(\rho(G))}{u'(\rho(G))}.
\end{align}
\end{proposition}
\begin{proof}
  By denoting $G_\epsilon(\cdot)$ the quantile function corresponding to the distribution function $F_\epsilon(\cdot)=(1-\epsilon)F(\cdot)+\epsilon \delta_z(\cdot)$, we have
    \baa
    \rho(G_\epsilon) &=& u^{-1}\left(\int_\mathbb{R} u(|x\wedge 0|)dF_\epsilon(x)\right)\\
    &=& u^{-1}\left((1-\epsilon)\int_\mathbb{R} u(|x\wedge 0|)dF(x)+\epsilon u(|z\wedge 0|)\right).
    \eaa
  Now, simple calculus leads to \eqref{eq:sensitivity_CE}.
\end{proof}

From Proposition \ref{prop:sensitivity_CE}, if $\lim_{x\rightarrow \infty}u(x)=\infty$, $\lim_{z\downarrow -\infty}S(z;G)=+\infty$, showing that the sensitivity function is unbounded. In particular, for the $L^p$ and entropic risk measures, the sensitivity functions are unbounded.

\subsection{Boundedness of sensitivity functions for robust risk estimators}
In this section, we compute the sensitivity functions of the
$\delta$-truncated versions of the loss certainty equivalents.
These truncated versions were introduced in Section
\ref{sec:robustification} in order to obtain robust risk
estimators. The conclusion of the following proposition is that by
truncating these risk measures, their sensitivity functions become
bounded. Therefore, the robustness properties of risk estimators
derived in Section \ref{sec:robustness} are consistent with the
sensitivity functions computed in this section.

\begin{proposition}\label{prop:sensitivity_robust_CE}
Consider the $\delta$-truncation of the loss certainty equivalent \eqref{eq:CEriskmeasuretruncated} with $0<\delta <1$. Assume $F(z)<1$ and $G(\cdot)$ is differentiable at $\delta$. Then, the sensitivity $S(z;G)$ can be  computed as follows:
\begin{enumerate}
  \item[\bf(i)] When $G(\delta)>0$,
\begin{align}\label{eq:SensitivityCEPositive}
  S(z;G) =  0.
\end{align}
\item[\bf(ii)] When $G(\delta)=0$,
\begin{align}\label{eq:SensitivityCEZero}
  S(z;G) =\frac{1}{u'(\rho_\delta(G))}\begin{cases}
    0, & z\ge G(\delta),\\
    \delta \, (1-\delta) \, G'(\delta), & z< G(\delta).
  \end{cases}
\end{align}
\item[\bf(iii)] When $G(\delta)<0$,
\begin{align}\label{eq:SensitivityCENegative}
  S(z;G) = \frac{1}{u'(\rho_\delta(G))}\left[-u(\rho_\delta(G)) + \begin{cases}
    u(|z\wedge 0|) - \delta^2 u'(|G(\delta)|) G'(\delta) , & z>G(\delta),\\
    u(|z\wedge 0|) , & z=G(\delta),\\
    u(|G(\delta)|) + \delta(1-\delta)u'(|G(\delta)|)G'(\delta), & z<G(\delta).
  \end{cases}\right]
\end{align}
\end{enumerate}
\end{proposition}
\begin{proof}
  First let us recall some properties of left-continuous inverse functions. For any distribution function $\widetilde F$, denote by $\widetilde F^{-1}$ its left-continuous inverse. Then, for any $t\in[0,1],x\in\mathbb{R}$,
  \begin{align}\label{eq:LeftInversePropertyOne}
  \widetilde F(x)\ge t\Longleftrightarrow x\ge \widetilde F^{-1}(t),\quad  \widetilde F(x)<t\Longleftrightarrow x< \widetilde F^{-1}(t),
  \end{align}
  and thus
  \begin{align}\label{eq:LeftInversePropertyTwo}
    \widetilde F(x-)<t\le \widetilde F(x)\Longleftrightarrow x = \widetilde F^{-1}(t).
  \end{align}
  \iffalse
  and
  \begin{align*}
  \widetilde F(x-)> t\Longleftrightarrow x> \widetilde F^{-1}(t+),\quad  \widetilde F(x-)\le t\Longleftrightarrow x\le  \widetilde F^{-1}(t+).
  \end{align*}\fi

  In the following, we denote by $F(\cdot)$ the distribution function associated with $G(\cdot)$. By assumption, $G(\cdot)$ is differentiable and thus continuous at $\delta$. We claim that
  \begin{align}\label{eq:LeftInversePropertyThree}
    F(z-)\le \delta \le F(z) \Longleftrightarrow z = G(\delta),\quad  F(z)>\delta \Longleftrightarrow z>G(\delta), \quad F(z)<\delta \Longleftrightarrow z<G(\delta).
  \end{align}
  Indeed, from \eqref{eq:LeftInversePropertyOne}, we deduce that $F(z)<\delta$ if only if $z<G(\delta)$. From \eqref{eq:LeftInversePropertyTwo}, we deduce that if $z = G(\delta)$, then $F(z-)\le \delta \le F(z)$. Thus, to prove \eqref{eq:LeftInversePropertyThree}, we only need to show that if $F(z-)\le \delta \le F(z)$, $z = G(\delta)$. Suppose $F(z-)\le \delta \le F(z)$. On the one hand, from \eqref{eq:LeftInversePropertyOne}, $z\ge G(\delta)$. On the other hand, for any $\varepsilon>0$ small enough, $F(z-\varepsilon)<\delta+\varepsilon$, which again by \eqref{eq:LeftInversePropertyOne} leads to $z-\epsilon<G(\delta+\varepsilon)$. Letting $\varepsilon\downarrow 0$, by the continuity of $G(\cdot)$ at $\delta$, we conclude that $z\le G(\delta)$.

  Denote by $G_\epsilon(\cdot)$ the quantile function corresponding to the distribution function $F_\epsilon(\cdot)=(1-\epsilon)F(\cdot)+\epsilon \delta_z(\cdot)$. We only need to compute
  \begin{align*}
  A:=\lim_{\epsilon\downarrow 0}\frac{\int_0^1u(|G_\epsilon(t\vee \delta)\wedge 0|)dt - \int_0^1u(|G(t\vee \delta)\wedge 0|)dt}{\epsilon},
  \end{align*}
  and $S(z;G)$ follows from the chain rule.

  Straightforward computation shows that
  \begin{align*}
  G_\epsilon(t) = \begin{cases}
    G\left(\frac{t-\epsilon}{1-\epsilon}\right), & t>\epsilon + (1-\epsilon)F(z),\\
    G\left(\frac{t}{1-\epsilon}\right), & t\le (1-\epsilon)F(z-),\\
    z, & (1-\epsilon)F(z-)<t\le \epsilon  + (1-\epsilon)F(z).
  \end{cases}
    \end{align*}
  Thus,
  \begin{align*}
  &\int_\delta^1u(|G_\epsilon(t)\wedge 0|)dt\\
   =& \int_\delta^1u(|G\left(\frac{t-\epsilon}{1-\epsilon}\right)\wedge 0|)\mathbf{1}_{\{t>\epsilon + (1-\epsilon)F(z)\}}dt + \int_\delta^1u(|G\left(\frac{t}{1-\epsilon}\right)\wedge 0|)\mathbf{1}_{\{t\le (1-\epsilon)F(z-)\}}dt \\
   &\quad + \int_\delta^1u(|z\wedge 0|)\mathbf{1}_{\{(1-\epsilon)F(z-)<t\le \epsilon  + (1-\epsilon)F(z)\}}dt\\
    =& (1-\epsilon)\int_{\frac{\delta-\epsilon}{1-\epsilon}}^1u(|G\left(t\right)\wedge 0|)\mathbf{1}_{\{t>F(z)\}}dt + (1-\epsilon)\int_{\frac{\delta}{1-\epsilon}}^1u(|G\left(t\right)\wedge 0|)\mathbf{1}_{\{t\le F(z-)\}}dt \\
   &\quad + \int_\delta^1u(|z\wedge 0|)\mathbf{1}_{\{(1-\epsilon)F(z-)<t\le \epsilon  + (1-\epsilon)F(z)\}}dt\\
   =&(1-\epsilon)\int_\delta^1u(|G\left(t\right)\wedge 0|)dt +(1-\epsilon)\int_{\frac{\delta-\epsilon}{1-\epsilon}}^{\delta}u(|G\left(t\right)\wedge 0|)\mathbf{1}_{\{t>F(z)\}}dt\\
   &\quad - (1-\epsilon)\int_{\delta}^{\frac{\delta}{1-\epsilon}}u(|G\left(t\right)\wedge 0|)\mathbf{1}_{\{t\le F(z-)\}}dt - (1-\epsilon)\int_\delta^1u(|G\left(t\right)\wedge 0|)\mathbf{1}_{\{F(z-)<t\le F(z)\}}dt\\
   &\quad  + \int_\delta^1u(|z\wedge 0|)\mathbf{1}_{\{(1-\epsilon)F(z-)<t\le \epsilon  + (1-\epsilon)F(z)\}}dt.
\end{align*}
It is then easy to show
\begin{align*}
  \int_{\frac{\delta-\epsilon}{1-\epsilon}}^{\delta}u(|G\left(t\right)\wedge 0|)\mathbf{1}_{\{t>F(z)\}}dt &= \begin{cases}
    0, & F(z)\ge \delta,\\
    (1-\delta)u(|G(\delta)\wedge 0|)\epsilon + o(\epsilon), & F(z)<\delta,
  \end{cases}\\
  \int_{\delta}^{\frac{\delta}{1-\epsilon}}u(|G\left(t\right)\wedge 0|)\mathbf{1}_{\{t\le F(z-)\}}dt &= \begin{cases}
    \delta u(|G(\delta)\wedge 0|)\epsilon + o(\epsilon), & F(z-)> \delta,\\
    0, & F(z-)\le \delta.
    \end{cases}
\end{align*}
Noticing \eqref{eq:LeftInversePropertyTwo}, we can show that
\begin{align*}
  \int_\delta^1u(|G\left(t\right)\wedge 0|)\mathbf{1}_{\{F(z-)<t\le F(z)\}}dt = \begin{cases}
    0, & F(z)< \delta,\\
    (F(z)-\delta)u(|z\wedge 0|), & F(z-)\le\delta\le F(z),\\
    (F(z)-F(z-))u(|z\wedge 0|), & \delta< F(z-).
  \end{cases}
\end{align*}
Similarly, we can show that
{\footnotesize
\begin{align*}
  \int_\delta^1u(|z\wedge 0|)\mathbf{1}_{\{(1-\epsilon)F(z-)<t\le \epsilon  + (1-\epsilon)F(z)\}}dt = \begin{cases}
    o(\epsilon), & F(z)< \delta,\\
    \left[(1-\epsilon)(F(z)-\delta)+\epsilon (1-\delta)\right]u(|z\wedge 0|), & F(z-)\le \delta\le F(z),\\
    \left[(1-\epsilon)(F(z)-F(z-))+\epsilon\right]u(|z\wedge 0|), & \delta<F(z-).
  \end{cases}
\end{align*}}
Therefore,
{\footnotesize
\begin{align*}
  &\int_\delta^1u(|G_\epsilon(t)\wedge 0|)dt=(1-\epsilon)\int_\delta^1u(|G\left(t\right)\wedge 0|)dt+ \begin{cases}
    \left[u(|z\wedge 0|) - \delta u(|G(\delta)\wedge 0|)\right]\epsilon + o(\epsilon), & \delta<F(z-),\\
    (1-\delta)u(|z\wedge 0|)\epsilon +o(\epsilon), & F(z-)\le\delta \le F(z),\\
    (1-\delta)u(|G(\delta)\wedge 0|)\epsilon + o(\epsilon), & \delta>F(z).
  \end{cases}
\end{align*}}

On the other hand
\begin{align*}
  G_\epsilon(\delta)\wedge 0 - G(\delta)\wedge 0 = \begin{cases}
    G\left(\frac{\delta-\epsilon}{1-\epsilon}\right)\wedge 0 - G(\delta)\wedge 0, & \delta >\epsilon + (1-\epsilon)F(z),\\
    G\left(\frac{\delta}{1-\epsilon}\right)\wedge 0- G(\delta)\wedge 0, & \delta\le (1-\epsilon)F(z-),\\
    z\wedge 0- G(\delta)\wedge 0, & (1-\epsilon)F(z-)<\delta \le \epsilon  + (1-\epsilon)F(z).
  \end{cases}
\end{align*}
We discuss case by case.
\begin{enumerate}
  \item $G(\delta)>0$. Because $G(\cdot)$ is continuous at $\delta$, it is easy to show that $G_\epsilon(\delta)\wedge 0 - G(\delta)\wedge 0=0$ when $\epsilon$ is sufficiently small.
  \item $G(\delta)=0$. In this case, we have
  \begin{align*}
  G_\epsilon(\delta)\wedge 0 - G(\delta)\wedge 0 = \begin{cases}
    G\left(\frac{\delta-\epsilon}{1-\epsilon}\right), & \delta >\epsilon + (1-\epsilon)F(z),\\
    0, & \delta\le (1-\epsilon)F(z-),\\
    z\wedge 0, & (1-\epsilon)F(z-)<\delta \le \epsilon  + (1-\epsilon)F(z)
  \end{cases}
\end{align*}
when $\epsilon$ is sufficiently small. Recalling \eqref{eq:LeftInversePropertyThree}, we conclude that
\begin{align*}
  G_\epsilon(\delta)\wedge 0 - G(\delta)\wedge 0 = \begin{cases}
    0, & \delta <F(z-),\\
    0, & F(z-)\le \delta\le F(z),\\
    -(1-\delta)G'(\delta)\epsilon + o(\epsilon), & \delta>F(z).
  \end{cases}
\end{align*}
\item $G(\delta)<0$. In this case, we have
\begin{align*}
  G_\epsilon(\delta)\wedge 0 - G(\delta)\wedge 0 = \begin{cases}
    G\left(\frac{\delta-\epsilon}{1-\epsilon}\right) - G(\delta), & \delta >\epsilon + (1-\epsilon)F(z),\\
    G\left(\frac{\delta}{1-\epsilon}\right)- G(\delta), & \delta\le (1-\epsilon)F(z-),\\
    z- G(\delta), & (1-\epsilon)F(z-)<\delta \le \epsilon  + (1-\epsilon)F(z).
  \end{cases}
\end{align*}
  when $\epsilon$ is small enough, leading to
  \begin{align*}
  G_\epsilon(\delta)\wedge 0 - G(\delta)\wedge 0 = \begin{cases}
    \delta G'(\delta)\epsilon + o(\epsilon), & \delta <F(z-),\\
    0, & F(z-)\le \delta\le F(z),\\
    -(1-\delta)G'(\delta)\epsilon + o(\epsilon), & \delta>F(z).
  \end{cases}
\end{align*}
\end{enumerate}

Notice that
\begin{align*}
  \int_0^1u(|G_\epsilon(t\vee \delta)\wedge 0|)dt - \int_0^1u(|G(t\vee \delta)\wedge 0|)dt = \int_\delta^1u(|G_\epsilon(t)\wedge 0|)dt - \int_\delta^1u(|G(t)\wedge 0|)dt\\
  + \delta \left(u(|G_\epsilon(\delta)\wedge 0|) - u(|G(\delta)\wedge 0|)\right).
\end{align*}
Then, when $G(\delta)>0$,
\begin{align*}
  &\lim_{\epsilon\downarrow 0}\frac{\int_0^1u(|G_\epsilon(t\vee \delta)\wedge 0|)dt - \int_0^1u(|G(t\vee \delta)\wedge 0|)dt}{\epsilon}\\
  =& \lim_{\epsilon\downarrow 0}\frac{\int_\delta^1u(|G_\epsilon(t)\wedge 0|)dt - \int_\delta^1u(|G(t)\wedge 0|)dt + \delta\left(u(|G_\epsilon(\delta)\wedge 0|) - u(|G(\delta)\wedge 0|)\right)}{\epsilon}\\
   =& \lim_{\epsilon\downarrow 0}\frac{\int_\delta^1u(|G_\epsilon(t)\wedge 0|)dt - \int_\delta^1u(|G(t)\wedge 0|)dt - \delta u'(|G(\delta)\wedge 0|)\left(G_\epsilon(\delta)\wedge 0 - G(\delta)\wedge 0\right)}{\epsilon}\\
  =& 0.
\end{align*}
Similarly, when $G(\delta) = 0$,
{\footnotesize
\begin{align*}
  \lim_{\epsilon\downarrow 0}\frac{\int_0^1u(|G_\epsilon(t\vee \delta)\wedge 0|)dt - \int_0^1u(|G(t\vee \delta)\wedge 0|)dt}{\epsilon}
  = \begin{cases}
    0, & \delta <F(z-),\\
    0, & F(z-)\le \delta\le F(z),\\
    u'(0)\delta (1-\delta) G'(\delta), & \delta>F(z).
  \end{cases}
\end{align*}}
When $G(\delta)<0$,
\begin{align*}
  &\lim_{\epsilon\downarrow 0}\frac{\int_0^1u(|G_\epsilon(t\vee \delta)\wedge 0|)dt - \int_0^1u(|G(t\vee \delta)\wedge 0|)dt}{\epsilon}\\
  =& -\int_\delta^1u(|G(t)\wedge 0|)dt + \begin{cases}
    u(|z\wedge 0|) - \delta u(|G(\delta)|) - \delta^2 u'(|G(\delta)|) G'(\delta) , & \delta <F(z-),\\
    (1-\delta)u(|z\wedge 0|) , & F(z-)\le \delta\le F(z),\\
    (1-\delta)u(|G(\delta)|) + \delta(1-\delta)u'(|G(\delta)|)G'(\delta), & \delta>F(z).
  \end{cases}\\
  =& -\int_0^1u(|G(t\vee \delta)\wedge 0|)dt + \begin{cases}
    u(|z\wedge 0|) - \delta^2 u'(|G(\delta)|) G'(\delta) , & \delta <F(z-),\\
    u(|z\wedge 0|) , & F(z-)\le \delta\le F(z),\\
    u(|G(\delta)|) + \delta(1-\delta)u'(|G(\delta)|)G'(\delta), & \delta>F(z).
  \end{cases}
\end{align*}

Finally, applying chain rule and \eqref{eq:LeftInversePropertyThree}, we immediately have \eqref{eq:SensitivityCEPositive}-\eqref{eq:SensitivityCENegative}.
\end{proof}

From Proposition \ref{prop:sensitivity_robust_CE}, it is easy to show that
\begin{align*}
  \sup_{z\in \mathbb{R}}S(z;G) \le \frac{1}{u'(\rho_\delta(G))}\times \begin{cases}
    0, & G(\delta)>0,\\
    \delta(1-\delta)G'(\delta), &G(\delta)=0,\\
    -u(\rho_\delta(G)) + u(|G(\delta)|) + \delta(1-\delta)u'(|G(\delta)|)G'(\delta), & G(\delta)<0.
  \end{cases}
\end{align*}
Thus, the truncated loss certainty equivalent  has bounded
sensitivity functions, which is consistent with the robustness
properties of risk estimators derived in Section
\ref{sec:robustness}.

\begin{remark}
  If $F$ has a continuous positive density $f$ in the neighborhood of $G(\delta)$, then  $G(\cdot)$ is differentiable at $\delta$ and
  \begin{align*}
    G'(\delta) = \frac{1}{f(G(\delta))}.
  \end{align*}
  Moreover, denoting by $\SVaR_\delta(z;G)$ the sensitivity function of $\VaR_\delta$ at $G$, which, from the proof of Proposition \ref{prop:sensitivity_robust_CE}, is
\baa
\SVaR_\delta(z;G) = \begin{cases}
    \frac{\delta}{f(G(\delta))} , & z> G(\delta),\\
    0, & z = G(\delta),\\
    -\frac{1-\delta}{f(G(\delta))}, & z< G(\delta),
  \end{cases}
\eaa we can rewrite, for $G(\delta)<0$, the sensitivity of the
$\delta$-truncation of the loss certainty equivalent  as: \baa
S(z;G) = \frac{1}{u'(\rho_\delta(G))}\left[-u(\rho_\delta(G)) +
u(|(z\vee G(\delta))\wedge 0|) - \delta
u'(|G(\delta)|)\SVaR_\delta(z;G)\right]. \eaa
\end{remark}

\section{Conclusions}\label{se:Conclusions}
In this paper, we have proposed a new class of risk measures named as loss-based risk measures, provided two representation theorems for convex loss-based risk measures, and investigated the robustness of the risk estimators associated with a family of statistical loss-based risk measures that include both statistical convex loss-based risk measures and VaR on losses as special cases.

Motivated by the fact that the risk measures employed in some of the risk management practices only depend on portfolio losses, we propose the loss-dependence property, a characterizing property of loss-based risk measures, which is largely overlooked and actually cannot be accommodated in the existing risk measure frameworks. We have shown in the paper a dual representation theorem for convex loss-based risk measures and another representation theorem if the risk measures are furthermore distributional-based. In addition, we have provided abundant interesting examples of loss-based risk measures, some of which cannot be obtained in the existing risk measure frameworks by simple modification.

In order to address the issue that risk estimates that are extremely sensitive to single data points are useless in some of the risk management practices, we have investigated the robustness of the risk estimators associated with a family of statistical loss-based risk measures. We have found a sufficient and necessary condition for those risk estimators to be robust, and this result significantly improves the existing ones. From that condition, we have shown that statistical convex loss-based risk measures lead to non-robust risk estimators while the loss-based VaR leads to a robust risk estimator, and these results have been confirmed by performing further sensitivity analysis. The conflict between the convexity property of a risk measure and the robustness of the corresponding risk estimator suggests that we need to decide which of the properties is more relevant when choosing a risk measure for certain use.

\appendix
\section{Lebesgue Continuity}\label{se:lebesgue}
In this section, we obtain a sufficient and necessary condition for the risk measure \eqref{eq:rhofunction} to satisfy the Lebesgue continuity defined in \citet{jouini2006law} and compare this condition with the one derived in Theorem \ref{th:weak_continuity}.
\begin{definition}[Lebesgue continuity]
Let $\rho$ be defined in \eqref{eq:rhofunction}. $\rho$ is {\em Lebesgue continuous} at $G\in \quan$ if for any $G_n\in \quan$ such that $G_n$ is uniformly bounded and $G_n\rightarrow G$, $\rho(G_n)\rightarrow \rho(G)$.
\end{definition}
 The following result gives the dual characterization of the
Lebesgue continuity and shows how the conclusion in Theorem \ref{th:weak_continuity}
is modified when the weak continuity is replaced by the Lebesgue continuity.
\begin{theorem}
  Let $\rho$ be defined as in \eqref{eq:rhofunction}. Then, the following are equivalent
  \begin{itemize}
  \item[(i)] $\rho(\cdot)$ is Lebesgue continuous at any $G\in\quan_c^\infty$.
  \item[(ii)] For each $c>0$,
  \begin{align*}
    \lim_{\delta\downarrow 0}\sup_{v(m)\le c}m((0,\delta))=0.
  \end{align*}
  \end{itemize}
\end{theorem}
\begin{proof}
We work with $U(\cdot):=-\rho(\cdot)$.
  \begin{itemize}
   \item[$(i)\Rightarrow(ii)$] If (ii) does not hold, there exists $c>0$, $\epsilon_0>0$, $\delta_n\downarrow 0$, and $m_n \in\cP((0,1))$ with $v(m_n)\le c$, such that $m_n((0,\delta_n))\ge \epsilon_0$. Define
       \begin{align*}
    G_n(z):=\begin{cases}
      -\beta, & 0<z\le\delta_n,\\
      \frac{\beta}{\delta_n}(z-2\delta_n),& \delta_n< z\le 2\delta_n,\\
      0, & 2\delta_n<z<1,
    \end{cases}
    \quad n\ge 1,
  \end{align*}
  where $\beta:=\frac{2c}{\epsilon_0}$. Then, we have
  \begin{align*}
    U(G_n(\cdot))&\le \int_{(0,1)}(G_n(z)\wedge 0)m_n(dz) + v(m_n)\\
        &\le -2c + c=-c<0.
  \end{align*}
  On the other hand, $G_n(\cdot)$ is uniformly bounded, and because $\delta_n\downarrow 0$, $G_n(\cdot)\rightarrow G(\cdot)\equiv0$. Because $U(0)=0$, $U(\cdot)$ is not Lebesgue continuous at $0$, which is a contradiction.

  \item[$(ii)\Rightarrow(i)$] Suppose $G,G_n\in\quan_c^\infty$ are uniformly bounded and $G_n\rightarrow G$. On the one hand,
  \begin{align*}
    \limsup_{n\rightarrow \infty}U(G_n(\cdot))&\le \inf_{m\in\cP((0,1))}\left\{\limsup_{n\rightarrow \infty}\int_{(0,1)}(G_n(z)\wedge 0)m(dz)+v(m)\right\}\\
    & = \inf_{m\in\cP((0,1))}\left\{\int_{(0,1)}(G(z)\wedge 0)m(dz)+v(m)\right\}\\
    & = U(G(\cdot)),
  \end{align*}
  where the first equality is due to bounded convergence theorem.

  On the other hand, let $M>0$ be a uniform bound of $G_n$, i.e., $\sup_{n\ge 1,z\in(0,1)}|G_n(z)|\le M$. For each fixed $0<\epsilon<1$, and $n\ge 1$, we can find $m_n\in\dom(v)$ such that
  \begin{align}\label{eq:UGn}
    U(G_n(\cdot))\ge \int_{(0,1)}(G_n(z)\wedge 0)m_n(dz) + v(m_n) -\epsilon.
  \end{align}
  Let $c_n:= v(m_n)<\infty$, then
  \begin{align*}
    c_n&\le U(G_n(\cdot))-\int_{(0,1)}(G_n(z)\wedge 0)m_n(dz)+\epsilon\\
    &\le 2M + 1=:c.
  \end{align*}
   Thus, we have
  \begin{align*}
    U(G_n(\cdot))\ge \inf_{v(m)\le c}\left\{\int_{(0,1)}(G_n(z)\wedge 0)m(dz) + v(m)\right\} -\epsilon,\quad n\ge 1.
  \end{align*}
  By monotonicity of $U$, we can find $0<\eta<1$ such that $U(G(\cdot\wedge \eta))\ge U(G(\cdot))-\epsilon$. Because $\lim_{\delta\downarrow 0}\sup_{v(m)\le c}m((0,\delta))=0$, we can find $0<\delta<\eta$ such that $\sup_{v(m)\le c}m((0,\delta))<\epsilon/M$. By Lemma \ref{le:uniformconvergence}, $G_n(z)\rightarrow G(z)$ uniformly on $[\delta,\eta]$. Thus, there exists $N$ such that for any $n\ge N$, $\sup_{z\in[\delta,1)}|G_n(z\wedge \eta)-G(z\wedge \eta)|<\epsilon$. Now, for each $n\ge N$, we have
  \begin{align*}
    U(G_n(\cdot))&\ge \inf_{v(m)\le c}\left\{\int_{(0,1)}(G_n(z\wedge \eta)\wedge 0)m(dz) + v(m)\right\} -\epsilon\\
    & =\inf_{v(m)\le c}\Big\{\int_{(0,1)}(G(z\wedge \eta)\wedge 0)m(dz) + \int_{(0,\delta)}[(G_n(z\wedge \eta)\wedge 0)-(G(z\wedge \eta)\wedge 0)]m(dz) \\
    &\quad + \int_{[\delta,1)}[(G_n(z\wedge \eta)\wedge 0)-(G(z\wedge \eta)\wedge 0)]m(dz)+ v(m)\Big\} -\epsilon\\
    &\ge \inf_{v(m)\le c}\Big\{\int_{(0,1)}(G(z\wedge \eta)\wedge 0 )m(dz) - 2M\cdot\frac{\epsilon}{M} -\epsilon + v(m)\Big\} -\epsilon\\
    & = \inf_{v(m)\le c}\Big\{\int_{(0,1)}(G(z\wedge \eta)\wedge 0)m(dz) + v(m)\Big\} -4\epsilon\\
    &\ge U(G(\cdot\wedge \eta))-4\epsilon\\
    &\ge U(G(\cdot))-5\epsilon.
  \end{align*}
Thus, we have $\liminf_{n\rightarrow \infty}U(G_n(\cdot))\ge U(G(\cdot))$. In summary, $U$ is Lebesgue continuous at $G$ and so is $\rho$.
\end{itemize}
\end{proof}

%\bibliography{reference}

%\iffalse

\end{document}